\documentclass[pdflatex,sn-basic,Numbered]{sn-jnl}

\usepackage[english]{babel}

\usepackage{indentfirst}
\usepackage{amsmath,amssymb,amsfonts}
\usepackage{graphicx}
\usepackage[title]{appendix}
\usepackage{todonotes}
\usepackage{placeins} 

\theoremstyle{thmstyleone}
\newtheorem{theorem}{Theorem}[section]
\newtheorem{proposition}[theorem]{Proposition}

\newtheorem{lemma}[theorem]{Lemma}

\begin{document}

\title[Exact Counts of Binary Phylogenetic Networks]
{Exact Counts of Binary Phylogenetic Networks with Four Reticulations}

\author[1]{\fnm{Hao} \sur{Yu}}

\author*[1]{\fnm{Louxin} \sur{Zhang}}
\email{matzlx@nus.edu.sg}

\affil[1]{
  \orgdiv{Department of Mathematics},
  \orgname{National University of Singapore},
  \orgaddress{
    \street{10 Lower Kent Ridge Road},
    \city{Singapore},
    \postcode{119076},
    \country{Singapore}}}

\abstract{
Phylogenetic networks provide a flexible framework for representing
reticulate evolutionary processes, such as hybridization, introgression,
recombination, and horizontal gene transfer. However, their combinatorial
complexity makes even basic enumeration problems difficult. Building on our
previous work for networks with up to three reticulations, we derive an
explicit closed-form formula for the number of unrestricted rooted binary
phylogenetic networks with four reticulations on \(n\) labeled taxa.

Our approach is based on tree-component graphs. We classify the 79 possible
component graphs corresponding to networks with four reticulations into
ten groups. We then enumerate the networks associated with each group by
combining known counts of one-component networks, forests, and networks with
fewer reticulations. Summing these contributions yields the desired
formula. This result extends the exact enumeration of unrestricted binary
phylogenetic networks to four reticulations and further demonstrates
the effectiveness of component graphs for systematically organizing and
counting increasingly complex network classes.
}


\maketitle

\section{Introduction}\label{section1}

Phylogenetic trees have long been used to represent the evolutionary relationships among a collection of taxa. However, these tree structures cannot adequately model evolutionary histories involving reticulation events, such as hybridization, introgression, recombination, and horizontal gene transfer. Phylogenetic network models are used to extend phylogenetic trees by allowing such events to be represented explicitly and therefore provide a more flexible framework for studying complex evolutionary histories (e.g., \cite{fontaine2015extensive,Huson_book,koblmuller2007reticulate}). However, the complexity of their structures makes their inference, enumeration, and exhaustive exploration extremely difficult. In particular, the number of possible networks grows rapidly with both the number of taxa and the number of reticulation events.

Many existing enumeration results focus on structurally restricted classes of
networks, such as tree-child networks, reticulation-visible networks, galled
networks, and tree-based networks
(e.g., \cite{batle2026exact,bouvel2020counting,cardona2020counting,
fuchs2021asymptotic,pons2021combinatorial}).
In particular, Pons and Batle \cite{pons2021combinatorial} conjectured an
identity for the enumeration of tree-child networks with a given number of
leaves and reticulations. This conjecture has recently been verified through
the work of Lin et al.\ \cite{lin2026proof}, together with earlier related
results. Yu and Zhang \cite{yu2026short} also gave a short direct combinatorial
proof of the Pons--Batle identity. 

In comparison, the counting and enumeration of unrestricted phylogenetic networks are much less developed even when the number of reticulation events is fixed and small. Closed-form formulas are known for networks with one reticulation (\cite{Zhang_19}) and two reticulations (\cite{mansouri2020counting}). In our previous work (\cite{HaoYu2026_JCB}), we derived a closed-form formula for networks with three reticulations using the component graph approach. We also provided a new rigorous proof of the first-order asymptotic formula for networks with any fixed number of reticulations, a result previously obtained in \cite{mansouri2020counting}. The present work continues this approach by treating the more complicated case of networks with four reticulations.

The rest of this paper is organized as follows. Section \ref{section2} introduces the necessary notation and definitions, reviews the component graph construction, and presents the counting formulas and identities used in the subsequent analysis. Section \ref{section3} classifies the component graphs with four reticulations into ten groups and computes the number of networks associated with each group. Summing these contributions yields the closed-form formula. Section \ref{section4} concludes this paper with some remarks on the enumeration and its possible applications. The proofs of identities in Lemma \ref{lem1} are provided in the Appendix.

\section{Preliminaries}\label{section2}

In this section, we introduce basic notation, concepts,  and identities that we will use later.

\subsection{Binary Phylogenetic Networks}

We use $\left[n\right]=\{1,2,...,n\}$ to denote the set of taxa. A \textit{binary phylogenetic network} (BPN) on $[n]$ is a rooted directed acyclic graph with no parallel edges satisfying the following conditions:
\begin{itemize}
    \item The {\it root} is of in-degree 0 and out-degree 1.
    \item There are $n$ labeled \textit{leaves} which are of in-degree 1 and out-degree 0, representing the taxa.
    \item The non-leaf and non-root nodes have either in-degree 1 and out-degree 2, or in-degree 2 and out-degree 1. The former are called \textit{tree nodes}, while the latter are called {\it reticulations}. 
    \item Edges are directed away from the root.
\end{itemize}

The edges entering reticulations are called \textit{reticulation edges}, and other edges (entering tree nodes or leaves) are called \textit{tree edges}. If $(u,v)$ is a directed edge, we say $u$ is a \textit{parent} of $v$ and $v$ is a \textit{child} of $u$. If there is a path starting from node $u$ to node $v$, $u$ is said to be an \textit{ancestor} of $v$ and to be \textit{above} $v$; $v$ is said to be a \textit{descendant} of $u$ and to be \textit{below} $u$. If two leaves $\ell_1$ and $\ell_2$ share the same parent, we say $\ell_1$ and $\ell_2$ form a \textit{cherry}.

\begin{proposition}\label{prp21}
    Each binary phylogenetic network with $k$ reticulations and $n$ leaves contains $2n+3k-1$ edges. 
\end{proposition}

A BPN containing no reticulations is called a \textit{binary phylogenetic tree}. A \textit{forest} of $m$ trees on $[n]$ is a collection of binary phylogenetic trees $\{{\cal T}_1,...,{\cal T}_m\}$ such that their leaf sets form a partition of $[n]$.

\begin{proposition}
\label{prp22}
(Proposition 2.8.1, \cite{semple_book}) Let $n\geq k$ and ${\cal F}_{n,k}$ denote the family of forests on $[n]$, each consisting of $k$ rooted trees. Then
\begin{eqnarray}\nonumber
|{\cal F}_{n, k}|=\frac{(2n-k-1)!}{2^{n-k}(n-k)!(k-1)!}.
\end{eqnarray}
In particular, the number $t_n$ of binary phylogenetic trees on $n$ taxa is given by $t_n=|{\cal F}_{n, 1}| = \frac{(2n-2)!}{2^{n-1}(n-1)!}$.
\end{proposition}

In the rest of this paper, binary phylogenetic networks are simply called \textit{networks}; binary phylogenetic trees are simply called \textit{trees}.

We will use ${\cal P}_{n,k}$ to denote the set of phylogenetic networks with $k$ reticulations on $[n]$.

\subsection{Component Graphs}

The key concept used in this work is the tree-component decomposition (e.g.,  \cite{cardona2020counting,gunawan2020counting}).

Removing all reticulation edges in an arbitrary network $N$ with $k$ reticulations results in a forest containing $k+1$ trees, whose roots are either the root or a reticulation of $N$. We call each obtained tree a \textit{tree-component} of $N$. Each tree-component may just be a single node or contain nodes with in-degree 1 and out-degree 1. Every vertex of $N$ belongs to exactly one tree-component.

The \textit{component graph} of a network is obtained by contracting each tree-component into a single vertex while preserving the connections induced by reticulation edges. Although the original network contains no parallel edges, its component graph may contain parallel edges. See Figure~\ref{fig:cg_example} for illustration.

\begin{figure}[h]
\centering
\includegraphics[width=0.5\linewidth]{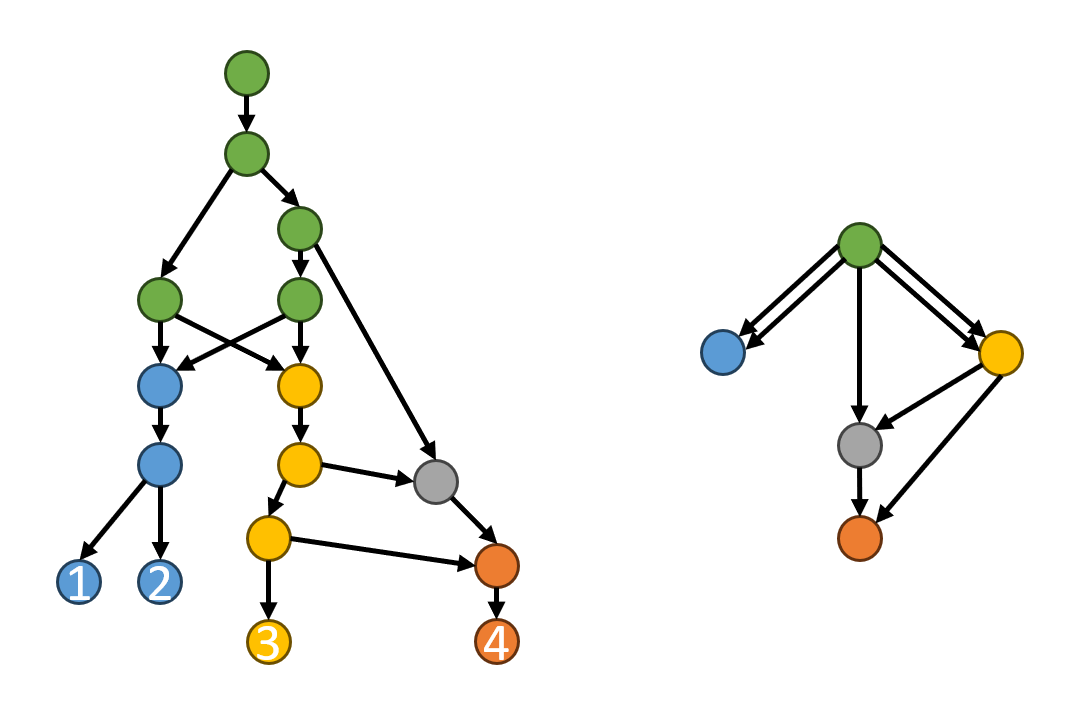}
\caption{\label{fig:cg_example} {\bf A network and its corresponding component graph}. The component graph consists of five vertices, each representing a tree-component of the network. Matching colors indicate the correspondence between tree-components and vertices of the component graph.}
\end{figure}

A network is said to be a \textit{one-component} network if the child of each reticulation is a leaf. We have the following established counting results for one-component networks and general networks with a small
number of reticulations.

\begin{proposition}\label{prp23}
    (1) (\cite{gunawan2020counting}) Let ${\cal OCP}_{n,k}$ denote the set of one-component networks with $k$ reticulations on $[n]$, where the $k$ children of reticulations are labeled $1,2,...,k$. Then 
    \begin{eqnarray}\nonumber
& |{\cal OCP}_{n,k}|=\left\{ 
 \begin{array}{ll}
 \frac{(2n-2)!}{2^{n-1}(n-2)!} & n>k=1,\\
 \frac{(n-1)^2 (2n-1)(2n-4)!}{2^{n-2}(n-2)!} & n\geq k=2, \\
 \frac{(4 n^5 - 16 n^4 + 17 n^3 - 7 n^2 + 2 n + 9) (2n-6)!}{2^{n-3}(n-3)!}
& n\geq k=3,\\
\frac{(8 n^7 - 44 n^6 + 38 n^5 + 63 n^4 + 24 n^3 + 26 n^2 - 502 n + 417)(2n-8)!}{2^{n-4}(n-4)!}
& n\geq k=4.
 \end{array} \right.
\end{eqnarray}

Note that ${\cal OCP}_{1,1}=\varnothing$, and hence
$\lvert {\cal OCP}_{1,1}\rvert=0$.

(2) (\cite{Zhang_19}) Let $n\ge 1$, 
\begin{equation}\label{eq1}
    |{\cal P}_{n,1}|=\frac{n(2n)!}{2^n n!}-2^{n-1}n!.
\end{equation}

(3) (\cite{mansouri2020counting}) Let $n\ge 1$,
\begin{equation}\label{eq2}
    |{\cal P}_{n,2}|=\frac{(2n-2)!}{3\cdot2^{n-1}(n-1)!}(6n^4+19n^3+18n^2-4n-6)-2^{n-1}(n+1)!(2n+3).
\end{equation}

(4) (\cite{HaoYu2026_JCB}) Let $n\ge 2$,
{
    \begin{eqnarray} 
    |{\cal P}_{n,3}|&=&\frac{(2n-2)!}{3(n-1)!2^n}(8n^6+88n^5+366n^4+640n^3+325n^2-155n-114) \nonumber\\
    && -\frac{1}{3}(n+1)!2^{n-4}(48n^3+367n^2+959n+840).
    \label{eq3}
    \end{eqnarray}
    }
Furthermore, $|{\cal P}_{n,3}|=9$ if $n=1$.

\end{proposition}

\subsection{Some Identities}
To derive a closed-form formula for $|{\cal P}_{n,4}|$, the following identities will be useful.

\begin{lemma}
\label{lem1}

(1)  For any $n\geq 1$, the following identities hold:
    
\begin{eqnarray}
    \sum_{k=0}^{n}\binom{2k}{k} \binom{2n-2k}{n-k}=2^{2n},  \label{eq4} \\
    \sum_{k=0}^{n} \binom{2k}{k} \binom{2n-2k}{n-k} \frac{1}{2k+1}
    =\frac{4^{2n}(n!)^2}{(2n+1)!},  \label{eq5}\\
     \sum_{k=0}^{n} \binom{2k}{k} \frac{1}{2^{2k}}=\binom{2n}{n}\frac{2n+1}{2^{2n}},\label{eq6}  \\
      \sum_{k=0}^{n} \binom{2n-2k}{n-k} \frac{2^{2k}}{2n-2k-1}=-\frac{(2n)!}{n!n!}.
    \label{eq7}   
\end{eqnarray}

(2) For any $n$ and $t$ such that  $0 \leq t\leq n-1$, the following identities hold: 
\begin{eqnarray}
    \sum_{k=0}^n \binom{2k}{k}\binom{2n-2k}{n-k}k(k-1)\cdots (k-t)=\binom{n}{t+1}2^{2n-t-1}(2t+1)!!,  \label{eq8}\\
     \sum_{k=0}^n \binom{2k}{k} \frac{k(k-1)\cdots (k-t)}{2^{2k}}=\frac{(2n+1)!}{2^{2n}n!(n-t-1)!(2t+3)}.
      \label{eq9}  
\end{eqnarray}

(3) For any $n$ and odd $t$ such that $1\leq t \leq 2n-1$, 

\begin{equation}\label{eq10}
    \sum_{k=0}^{n} \binom{2k}{k} \binom{2n-2k}{n-k} \frac{1}{2k-t}
    =0.
\end{equation}
\end{lemma}

The above identities can be proved using induction, hypergeometric functions \cite{diekema2022combinatorial} or the generating function 
$(1-4x)^{-1/2}=\displaystyle\sum_{k\geq 0}{2k\choose k}x^k$. Their proofs appear in the Appendix. 

\section{Counting Networks with Four Reticulations}\label{section3}

Figure~9 of \cite{cardona2020counting} displays 82 drawings of
component graphs with five vertices. However, three pairs of
drawings in that figure are isomorphic as rooted directed multigraphs.
After identifying each of these pairs as a single isomorphism class, we
obtain 79 non-isomorphic component graphs for networks with four
reticulations.

We partition these 79 component graphs into ten groups according to the
configuration of the edges leaving the top vertex; see
Figures~\ref{fig:41}, \ref{fig:42}, and~\ref{fig:43}. We compute the
number of networks associated with each group in order to derive a
closed-form formula for \(\lvert{\cal P}_{n,4}\rvert\). Let \(C_i\)
denote the number of networks associated with the \(i\)-th group.
All sums below are simplified using the identities in
Lemma~\ref{lem1}. The resulting algebraic simplifications can also be
verified using Maple.

In each decomposition below, we use the term \textit{top subnetwork} for
the subnetwork constructed from the top tree-component together with the
reticulations immediately below it and their designated children. The
remaining part is called the \textit{bottom subnetwork}.

In the counting arguments below, we use the term \textit{network leaf} for a leaf carrying a label from \([n]\), in order to distinguish such leaves from the auxiliary designated leaves or designated children used as attachment points in the decompositions.

To avoid confusion, in the decompositions below, when ${\cal OCP}_{m,r}$ is used, the labels on the children of the $r$ reticulations are regarded
only as temporary labels distinguishing designated attachment leaves.
These designated leaves will subsequently be identified with the roots of
the corresponding bottom subnetworks.

Throughout this section, the letters $S$ and $A$ in notations such
as $S_{\mathrm{top}}$, $A_{\mathrm{top}}$, $S_{\mathrm{low}}$, and
$A_{\mathrm{low}}$ distinguish the symmetric cases from the
remaining asymmetric cases, respectively; the subscripts indicate
the top and bottom subnetworks.

We occasionally use the terms \textit{network} and \textit{one-component network} in quotation marks for intermediate structures that satisfy the corresponding degree conditions but may contain parallel edges.

\begin{figure}[t!]
\centering
\includegraphics[width=0.8\linewidth]{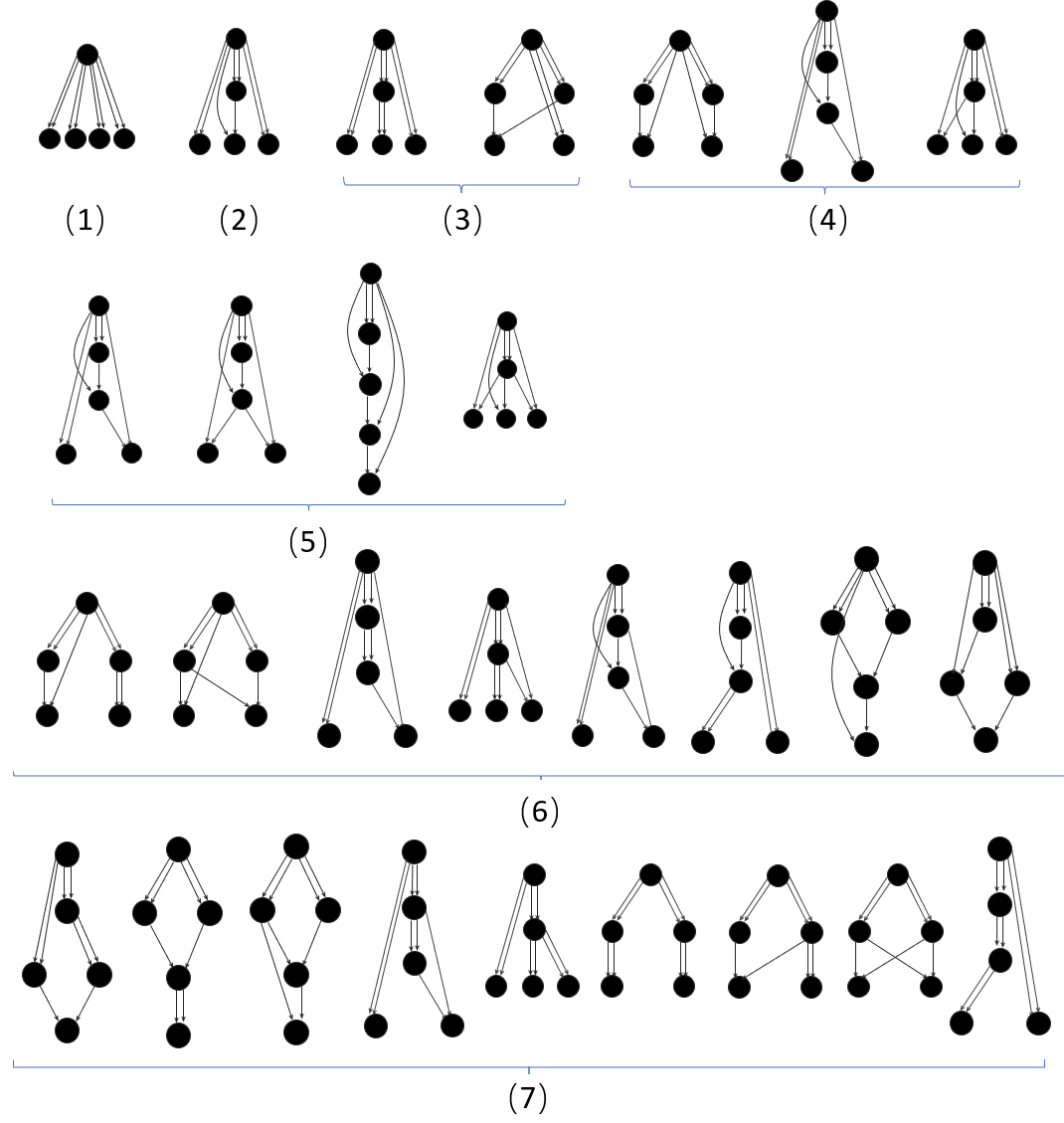}
\caption{\label{fig:41} \bf Component graphs for the first seven groups.}
\end{figure}

\begin{figure}[t!]
\centering
\includegraphics[width=0.7\linewidth]{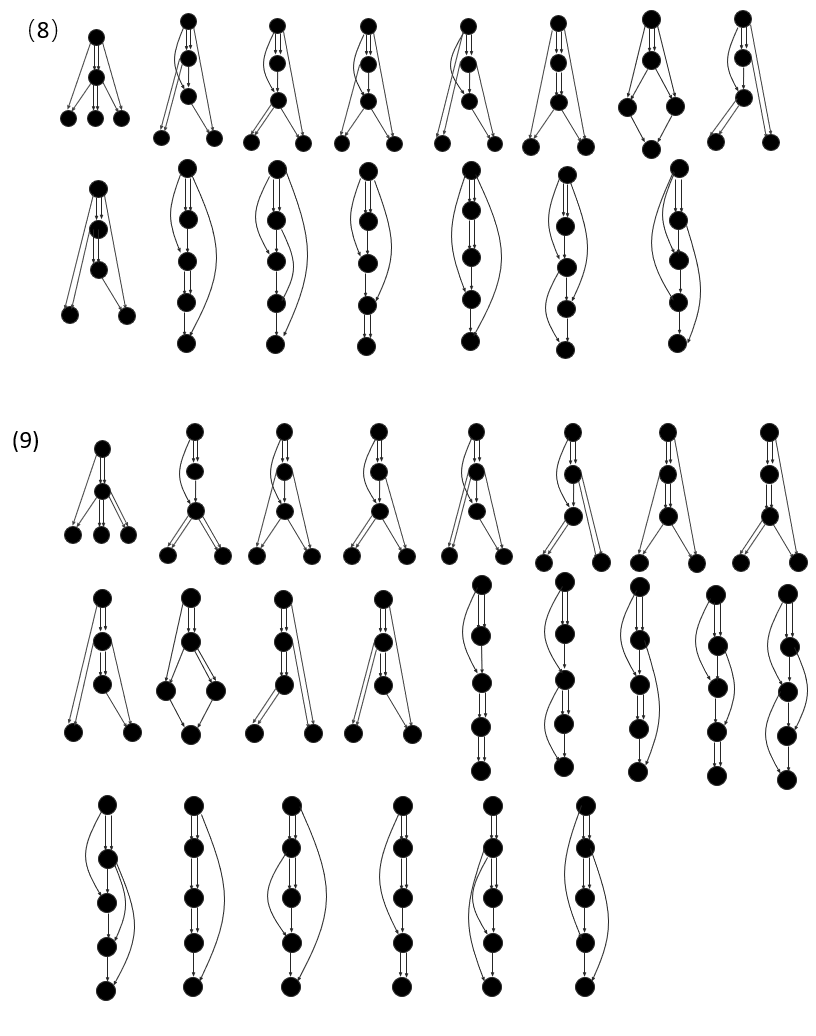}
\caption{\label{fig:42} \bf Component graphs for the eighth and ninth groups.}
\end{figure}

\begin{figure}[t!]
\centering
\includegraphics[width=0.7\linewidth]{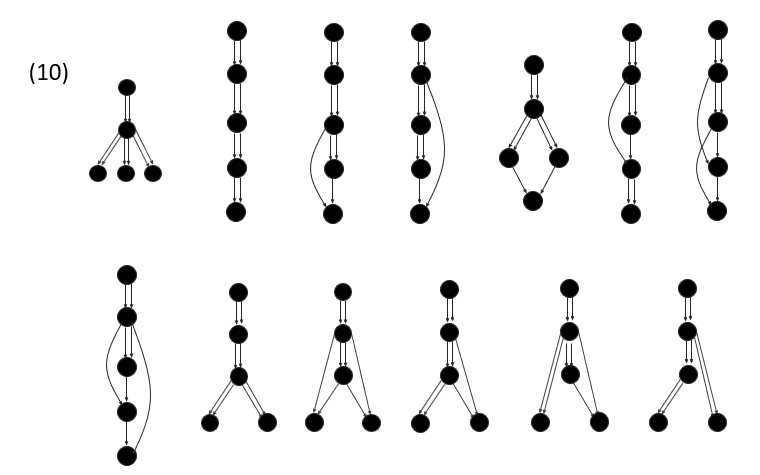}
\caption{\label{fig:43} \bf Component graphs for the tenth group.}
\end{figure}
\FloatBarrier

\begin{lemma}
\label{lem31}
    Let $n \geq 4$. The number of  networks in  ${\cal P}_{n,4}$ associated with the first group of component graphs in Figure \ref{fig:41} is:
    \begin{equation*}
\resizebox{0.98\linewidth}{!}{$\displaystyle
\begin{aligned}
C_1={}&
\frac{(2n-4)!}{3\cdot (n-2)!2^{n-1}}
\bigl(
8n^9+476n^8+6066n^7+30757n^6+60389n^5-13076n^4
-185488n^3-148465n^2+42471n+37494
\bigr)
\\
&-\frac{1}{3}n!2^{n-9}
\bigl(
3003n^6+75537n^5+688695n^4+3045455n^3
+6953742n^2+7718368n+3196224
\bigr).
\end{aligned}
$}
\end{equation*}
\end{lemma}

\begin{proof}
    Networks whose component graph belongs to the first group in Figure \ref{fig:41} have five tree-components. The top subnetwork is a one-component network with four reticulations, while the four remaining tree-components form a forest of four trees. Therefore, such a network can be obtained from a one-component network with four reticulations by replacing the children of the reticulations with the four trees from a forest. See Figure \ref{4_ret_fig_1}(a) for illustration. Hence, by Propositions \ref{prp23}(1) and \ref{prp22}, we calculate the number of such networks as follows.
{\footnotesize
    \begin{eqnarray*}
        C_1 &=& \sum_{j=0}^{n-4} \left[ \binom{n}{j} \times |{\cal F}_{n-j,4}| \times |{\cal OCP}_{j+4,4}|\right]\\
        &=& \sum_{j=0}^{n-4} \Bigg\{\frac{n!}{j!(n-j)!} \frac{(2n-2j-5)!}{6\cdot2^{n-j-4}(n-j-4)!}\Bigg[\frac{(j+4)^3(j+6)(2j+6)!}{2^{j+3}(j+3)!}\\&&+\frac{[8(j+4)^3-13(j+4)^2+8(j+4)-3]}{2^{j+3}(j+2)!}(2j+4)!\\
        &&+\frac{[12(j+4)^3-114(j+4)^2+312(j+4)-264]}{2^{j+3}(j+1)!}(2j+2)!+\frac{15(2j)!}{2^{j+1}j!} \Bigg] \Bigg\}\\
        &=& \frac{n!}{3\cdot2^n} \sum_{j=0}^{n-4} \binom{2j}{j} \binom{2n-2j}{n-j} \frac{(2j+1)(2j+3)(2j+5)(j+6)(j+4)^3(n-j-3)}{(2n-2j-1)(2n-2j-3)}\\
        &&+\frac{n!}{3\cdot2^{n+1}}\sum_{j=0}^{n-4} \binom{2j}{j} \binom{2n-2j}{n-j} \frac{(2j+1)(2j+3)(n-j-3)(8j^3+83j^2+288j+333)}{(2n-2j-1)(2n-2j-3)}\\
        &&+\frac{n!}{3\cdot2^{n+2}}\sum_{j=0}^{n-4} \binom{2j}{j} \binom{2n-2j}{n-j} \frac{(2j+1)(n-j-3)(12j^3+30j^2-24j-72)}{(2n-2j-1)(2n-2j-3)}\\
        &&+\frac{5n!}{2^{n+1}}\sum_{j=0}^{n-4} \binom{2j}{j} \binom{2n-2j}{n-j} \frac{n-j-3}{(2n-2j-1)(2n-2j-3)}\\
                &=&
        \resizebox{\dimexpr\linewidth-4.5em\relax}{!}{$\displaystyle
        \begin{aligned}
        &\frac{(2n-4)!}{3\cdot (n-2)!2^{n-1}}
        \Bigl(
        8n^9+476n^8+6066n^7+30757n^6+60389n^5-13076n^4
        -185488n^3-148465n^2+42471n+37494
        \Bigr)
        \\
        &\quad
        -\frac{1}{3}n!2^{n-9}
        \Bigl(
        3003n^6+75537n^5+688695n^4+3045455n^3
        +6953742n^2+7718368n+3196224
        \Bigr).
        \end{aligned}
        $}
    \end{eqnarray*}}
\end{proof}

\begin{lemma}
\label{lem32}
    Let $n \geq 3$. The number of  networks in  ${\cal P}_{n,4}$ associated with the second group of component graphs in Figure \ref{fig:41} is:
    \begin{eqnarray*}
        C_2&=&2^{n-8}n!(231n^6+8631n^5+91665n^4+439205n^3+1052520n^2+1204420n+508416)\\
    &&-\frac{(2n-1)!}{2^{n-1}(n-1)!}(12n^6+220n^5+1547n^4+5324n^3+9275n^2+7457n+1986).
    \end{eqnarray*}
\end{lemma}
\begin{proof}
    A network whose component graph belongs to the second group in Figure \ref{fig:41} can be divided into top and bottom subnetworks. The top subnetwork is a one-component network with three reticulations, one designated leaf $\ell$, and $j$ network leaves, where $j\geq0$. The three designated children of the reticulations in the top subnetwork are identified with the roots of the three trees in the bottom subnetwork, which can be regarded as a forest of three trees with $n-j$ leaves. The designated leaf $\ell$ from the top subnetwork is inserted into an arbitrary one of the edges of the forest, which provides $2n-2j-3$ ways. See Figure \ref{4_ret_fig_1}(b) for illustration. By Propositions \ref{prp22} and \ref{prp23}(1), we obtain:
    \begin{eqnarray*}
     C_2&=& \sum_{j=0}^{n-3} \binom{n}{j} |{\cal OCP}_{j+4,3}|\times|{\cal F}_{n-j,3}|\times(2n-2j-3)\\
     &=& \sum_{j=0}^{n-3} \frac{n!}{j!(n-j)!}\times\Big[\frac{(j+4)^3(2j+6)!}{2^{j+3}(j+3)!}+\frac{(j+4)(j+3)(2j+4)!}{2^{j+2}(j+2)!}-\frac{3(j+1)(2j+2)!}{2^{j+1}(j+1)!} \Big]\\
     &&\times \frac{(2n-2j-4)!}{2^{n-j-3}(n-j-3)!\times2} (2n-2j-3)\\
     &=& \frac{n!}{2^n} \sum_{j=0}^{n-3} \binom{2j}{j} \binom{2n-2j}{n-j} \frac{(2j+1)(2j+3)(2j+5)(n-j-2)(j+4)^3}{2n-2j-1}\\
     &&+\frac{n!}{2^n} \sum_{j=0}^{n-3} \binom{2j}{j} \binom{2n-2j}{n-j} \frac{(2j+1)(2j+3)(j+3)(j+4)(n-j-2)}{2n-2j-1}\\
     &&-\frac{3\cdot n!}{2^n}\sum_{j=0}^{n-3} \binom{2j}{j} \binom{2n-2j}{n-j} \frac{(2j+1)(j+1)(n-j-2)}{2n-2j-1}
    \\ &=& 2^{n-8}n!(231n^6+8631n^5+91665n^4+439205n^3+1052520n^2+1204420n+508416)\\
   &&-\frac{(2n-1)!}{2^{n-1}(n-1)!}(12n^6+220n^5+1547n^4+5324n^3+9275n^2+7457n+1986).
 \end{eqnarray*}
    
\end{proof}

\begin{figure}[b!]
    \centering
    \includegraphics[width=0.6\linewidth]{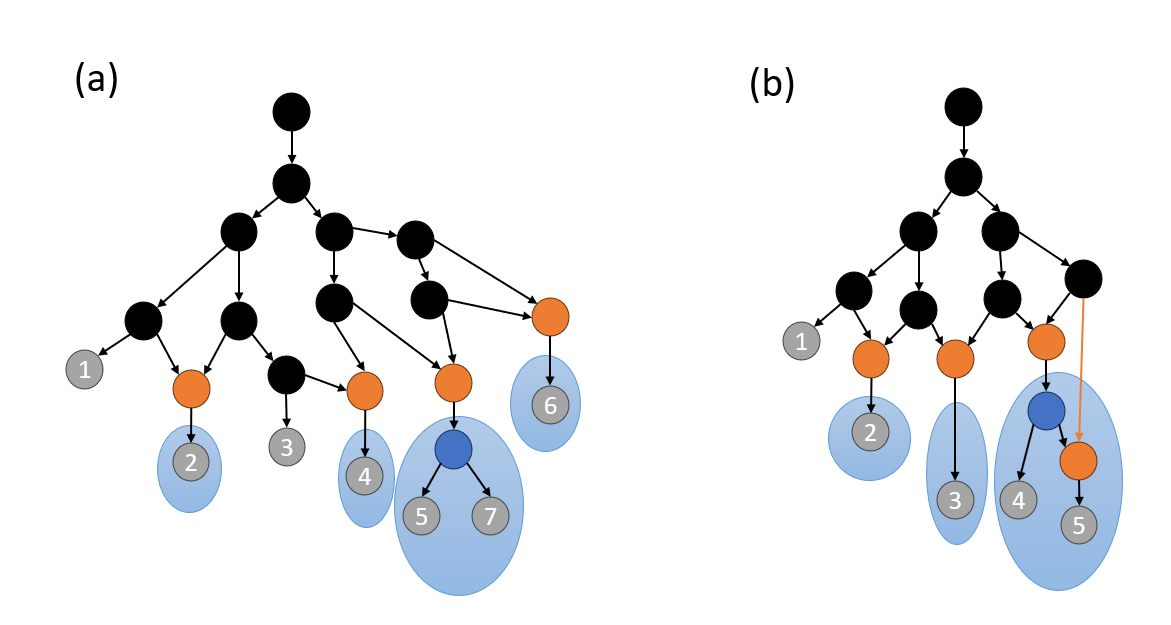}
    \caption{{\bf Illustration of the network decompositions in the proofs of Lemmas \ref{lem31} and \ref{lem32}.} The unshaded subnetwork is the top subnetwork, while the shaded subnetwork is the bottom subnetwork $P$ in each case. (a) $P$ is a forest of four trees. (b) $P$ is a forest of three trees, and we need to insert the designated leaf into an arbitrary edge of the forest.
    \label{4_ret_fig_1}}
    \end{figure}

\begin{lemma}
\label{lem33}
    Let $n \geq 2$. The number of networks in ${\cal P}_{n,4}$ associated with the third group of component graphs in Figure \ref{fig:41} is:
    \begin{equation*}
\resizebox{0.98\linewidth}{!}{$\displaystyle
\begin{aligned}
C_3={}&
n!2^{n-9}
\left(
21n^6+147n^5-2175n^4-21635n^3
-70854n^2-97312n-46272
\right)
\\
&-\frac{(2n-2)!}{231(n-1)!2^{n-1}}
\left(
84n^7-532n^6-13383n^5-64732n^4
-121380n^3-66424n^2+29823n+21714
\right).
\end{aligned}
$}
\end{equation*}
\end{lemma}

\begin{proof}
    There are two component graphs in the third group. We discuss them one by one. 

    For the left-hand component graph, a network associated with this component graph can be divided into top and bottom subnetworks. The top subnetwork is a one-component network with three reticulations. The bottom subnetwork is a network with one reticulation together with a forest of two trees. The network with one reticulation and the two trees can be regarded as the three bottom structures attached to the children of the three reticulations in the top subnetwork. See Figure~\ref{4_ret_fig_2} for an illustration. Let $P_1$ denote the number of networks corresponding to this component graph. Then we obtain:
    \begin{eqnarray*}
        P_1 &=& \sum_{j=0}^{n-4}\left\{ \binom{n}{j} |{\cal OCP}_{j+3,3}| \sum_{i=2}^{n-j-2}\left[ \binom{n-j}{i} |{\cal P}_{i,1}|\times |{\cal F}_{n-j-i,2}|\right]\right\}\\
        &=& \sum_{j=0}^{n-4}\Bigg\{ \frac{n!}{(n-j)!j!} \left[\frac{(j+3)^3(2j+4)!}{2^{j+2}(j+2)!}+\frac{(j+3)(j+2)(2j+2)!}{2^{j+1}(j+1)!}-\frac{3j(2j)!}{2^jj!} \right]\\
        &&\times\sum_{i=2}^{n-j-2} \left[\frac{(n-j)!}{i!(n-j-i)!}\left(\frac{i(2i)!}{2^i i!}-2^{i-1}i! \right)\times\frac{(2n-2j-2i-2)!}{2^{n-j-i-1}(n-j-i-1)!}  \right]\Bigg\}\\
        &=& \sum_{j=0}^{n-4} \frac{n!}{(n-j)!j!} \frac{(2j)!}{2^{j}j!}(4j^5+44j^4+185j^3+362j^2+311j+87)\\
        &&\times\left[\frac{(n-j)(n-j+1)(2n-2j-3)!}{2^{n-j-2}(n-j-2)!}-3(n-j)!2^{n-j-2} \right]\\
        &=&\frac{n!}{2^n} \sum_{j=0}^{n-4}\binom{2j}{j}\binom{2n-2j}{n-j} \frac{(n-j)(n-j+1)}{2n-2j-1}(4j^5+44j^4+185j^3+362j^2+311j+87)\\
        &&-3n!2^{n-2}\sum_{j=0}^{n-4}\binom{2j}{j} \frac{(4j^5+44j^4+185j^3+362j^2+311j+87)}{2^{2j}}\\
        &=&2^{n-9}n!(21n^6+651n^5+4545n^4+13285n^3+17466n^2+8000n+576)\\
        &&-\frac{(2n-3)!}{77\cdot(n-2)!2^{n-2}}n(84n^6+1008n^5+4173n^4+7263n^3+3745n^2-1590n-1285).
    \end{eqnarray*}

    \begin{figure}[b!]
    \centering
    \includegraphics[width=0.3\linewidth]{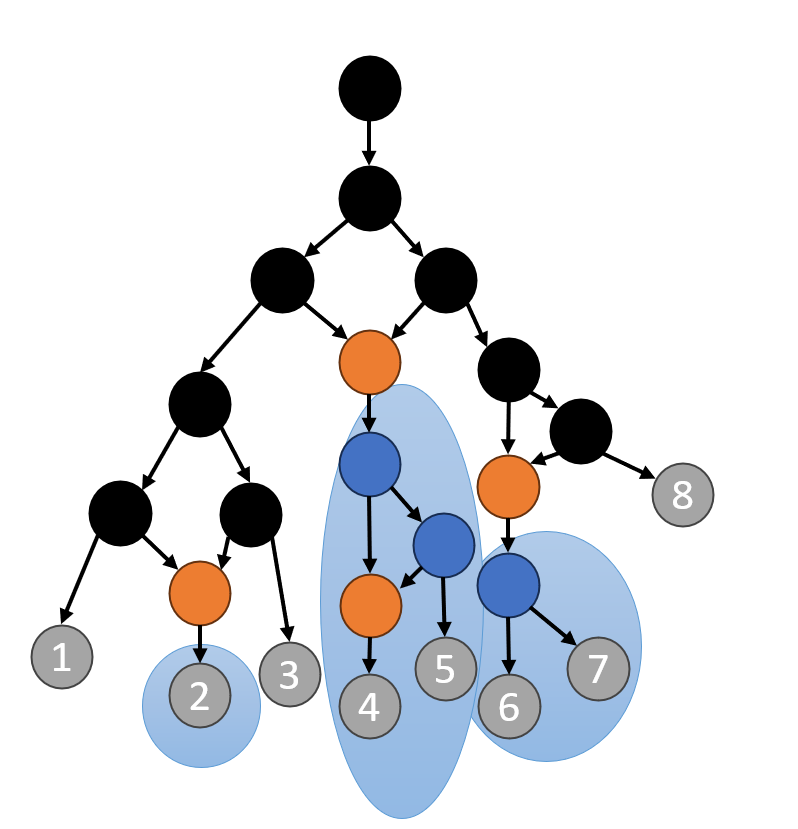}
    \caption{{\bf Illustration of the network decomposition for the component graph on the left-hand side of the third group in the proof of Lemma \ref{lem33}.} The unshaded subnetwork is the top subnetwork, while the shaded subnetwork is the bottom subnetwork. The bottom subnetwork is a network with one reticulation together with a forest of two trees.
    \label{4_ret_fig_2}}
    \end{figure}
    
The right-hand component graph is the most difficult case. A network associated with this component graph can be divided into two subnetworks. The top subnetwork is a one-component network with three reticulations, each having a designated child, and $j$ network leaves. The bottom subnetwork can be regarded as a set of two `subnetworks'. One of them can be regarded as a tree containing $i$ network leaves, which is denoted by ${\cal T}_1$, since there are no edges leaving the corresponding component. The remaining subnetwork can first be regarded as a forest of two trees $\{{\cal T}_2,{\cal T}_3\}$ with $n-j-i$ network leaves, where ${\cal T}_2$ contains a designated leaf $\ell$. Then we need to insert $\ell$ into an edge of ${\cal T}_3$. 


    First, the top subnetwork is a one-component network with three reticulations. If it contains $j$ network leaves, it will have $|{\cal OCP}_{j+3,3}|$ possible topologies. Denote the children of these reticulations by $a,b$ and $c$. Among these topologies, some are \textit{symmetric}, where switching $a$ and $b$ results in the same subnetwork. A direct case analysis of all possible identifications among the
parents of the reticulations above $a$ and $b$ yields exactly four
types of symmetric topologies, represented by the nine subcases shown
in Figure \ref{4_ret_fig_3}. Thus, the list in that figure is exhaustive.

    The two reticulations (the parents of $a$ and $b$) may share the same two parents $p_1$ and $p_2$, which are siblings. Denote the common parent of $p_1$, $p_2$ by $p$. To generate such a structure, $p$ can be inserted into an edge of a one-component network containing one reticulation with $j$ network leaves (Figure \ref{4_ret_fig_3}a). The edge leaving the unique reticulation cannot be selected, since otherwise the resulting network would not be one-component. This construction provides $|{\cal OCP}_{j+1,1}|\times(2j+3)$ possibilities. Also, it can be constructed from a `one-component network', which contains a pair of parallel edges, by inserting $p$ into one of the parallel edges. To generate such topologies, we can replace the edge to the leaf $j+1$ from a tree over $j+1$ taxa with the parallel-edge structure (Figure \ref{4_ret_fig_3}b), which provides $t_{j+1}$ possibilities.

    Consequently, the total number of top subnetworks containing such symmetric topologies is given by

    $$S_{top,1}(j)=|{\cal OCP}_{j+1,1}|\times(2j+3)+t_{j+1}=\frac{(2j)!}{2^j j!}(2j^2+3j+1). $$

    The two reticulations (the parents of $a$ and $b$) may share the same two parents $p_1,p_2$, whose parents $p_1',p_2'$ are distinct. To generate such structures, \(p_1'\) and \(p_2'\) can be inserted into either the same edge or two distinct edges of a one-component network with one reticulation and \(j\) network leaves, excluding the edge leaving the reticulation (Figure \ref{4_ret_fig_3}c). There are $\left[(2j+3)+\binom{2j+3}{2} \right]\times |{\cal OCP}_{j+1,1}|$ possibilities. Also, the two parents of reticulations can be inserted into a network containing parallel edges. Using the same idea, we can first replace the edge to the leaf $j+1$ from a tree over $j+1$ taxa with the parallel-edge structure. Then, we insert $p_1'$ and $p_2'$. If they are inserted into the same parallel edge (Figure \ref{4_ret_fig_3}d), it provides $t_{j+1}$ possibilities. If they are inserted into both parallel edges (Figure \ref{4_ret_fig_3}e), it provides $t_{j+1}$ possibilities. If one of them is inserted into one parallel edge, while the other is inserted into an original tree edge (Figure \ref{4_ret_fig_3}f), it provides $(2j+1)t_{j+1}$ possibilities.

    Consequently, the total number of top subnetworks containing such symmetric topologies is given by

    $$S_{top,2}(j)=\left[(2j+3)+\binom{2j+3}{2} \right]\times |{\cal OCP}_{j+1,1}|+(2j+3)t_{j+1}=\frac{(2j)!}{2^jj!}(2j^3+7j^2+8j+3). $$

    The two reticulations (the parents of $a$ and $b$) may have three distinct parents in total (Figure \ref{4_ret_fig_3}g). That is, one parent $p$ is the same, while the other two parents $p_1$ and $p_2$ are distinct. After deleting edges from these three parents and contracting all nodes with in-degree 1 and out-degree 1, the remaining subnetwork is just `a one-component network' with one reticulation and a pair of parallel edges. This structure can be obtained by inserting $p_1$ and $p_2$ into both parallel edges, and inserting $p$ into another tree edge.

    Consequently, the total number of top subnetworks containing this type of symmetric topology is given by

    $$S_{top,3}(j)=(2j+1)t_{j+1}=\frac{(2j)!}{2^jj!}(2j+1). $$

    The remaining possible symmetric topologies are illustrated in
Figures~\ref{4_ret_fig_3}h and~\ref{4_ret_fig_3}i. For \(j\geq1\), such a
subnetwork is obtained by inserting one of the two structures into an edge
of a tree on \(j\) network leaves. Hence
\[
S_{\mathrm{top},4}(j)
   =2(2j-1)t_j
   =\frac{(2j)!}{2^{j-1}j!},
   \qquad j\geq1.
\]
When \(j=0\), the two structures in
Figures~\ref{4_ret_fig_3}h and~\ref{4_ret_fig_3}i give
\(S_{\mathrm{top},4}(0)=2\), which agrees with the closed expression on the
right. Therefore, for every \(j\geq0\),
\[
S_{\mathrm{top},4}(j)=\frac{(2j)!}{2^{j-1}j!}.
\]

    \begin{figure}[t!]
\centering
\includegraphics[width=0.5\linewidth]{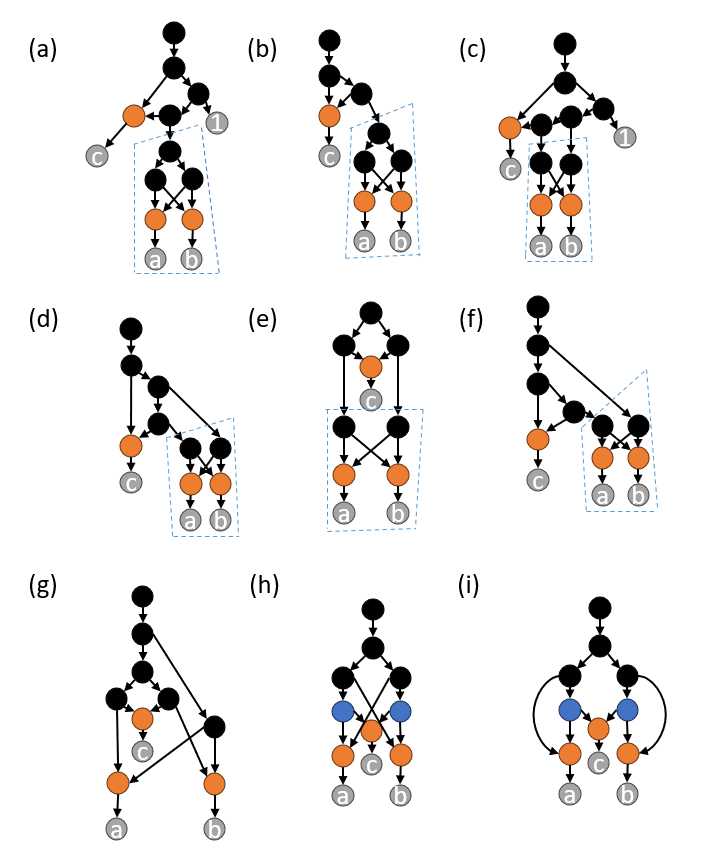}
\caption{\label{4_ret_fig_3} {\bf Symmetric topologies of the top subnetwork in a network with the component graph on the right-hand side in the third group.} (a)-(b) {\bf First symmetric topology:} the two common parents of the reticulations are siblings. (c)-(f) {\bf Second symmetric topology:} the two common parents of the reticulations are not siblings. (g) {\bf Third symmetric topology:} there are three parents of reticulations in total, two of which are from original parallel edges. (h)-(i) {\bf Fourth symmetric topology:} structures containing all three reticulations. }
\end{figure}
    
    Consequently, for $j \geq 0$, there are
    \begin{eqnarray*}
S_{top}(j)&=&S_{top,1}(j)+S_{top,2}(j)+S_{top,3}(j)+S_{top,4}(j) \\
&=&\frac{(2j)!}{2^jj!}\left[(2j^2+3j+1)+(2j^3+7j^2+8j+3)+(2j+1)+2 \right]\\
&=&\frac{(2j)!}{2^jj!}(2j^3+9j^2+13j+7)        
    \end{eqnarray*}
    possible symmetric topologies and
    \[
\begin{aligned}
A_{\mathrm{top}}(j)
&=
\lvert{\cal OCP}_{j+3,3}\rvert-S_{\mathrm{top}}(j)
\\
&=
\frac{(j+3)^3(2j+4)!}{2^{j+2}(j+2)!}
+\frac{(j+3)(j+2)(2j+2)!}{2^{j+1}(j+1)!}
\\
&\quad
-(2j^3+9j^2+16j+7)\frac{(2j)!}{2^j j!}
\end{aligned}
\]
    possible asymmetric topologies for the top subnetwork.

    Then we turn our focus to the bottom subnetwork, which also has a symmetric case. If ${\cal T}_2$ only contains $\ell$ and we insert $\ell$ into the edge from the root of ${\cal T}_3$, we will get a symmetric bottom topology (Figure \ref{4_ret_fig_4}a, where $\ell$ is the orange node). Then the part below $\ell$ is a tree over $n-j-i$ network leaves. We have

    $$S_{low}(n-j-i)=t_{n-j-i}=\frac{(2n-2j-2i-2)!}{2^{n-j-i-1}(n-j-i-1)!}.$$

    \begin{figure}[t!]
\centering
\includegraphics[width=0.5\linewidth]{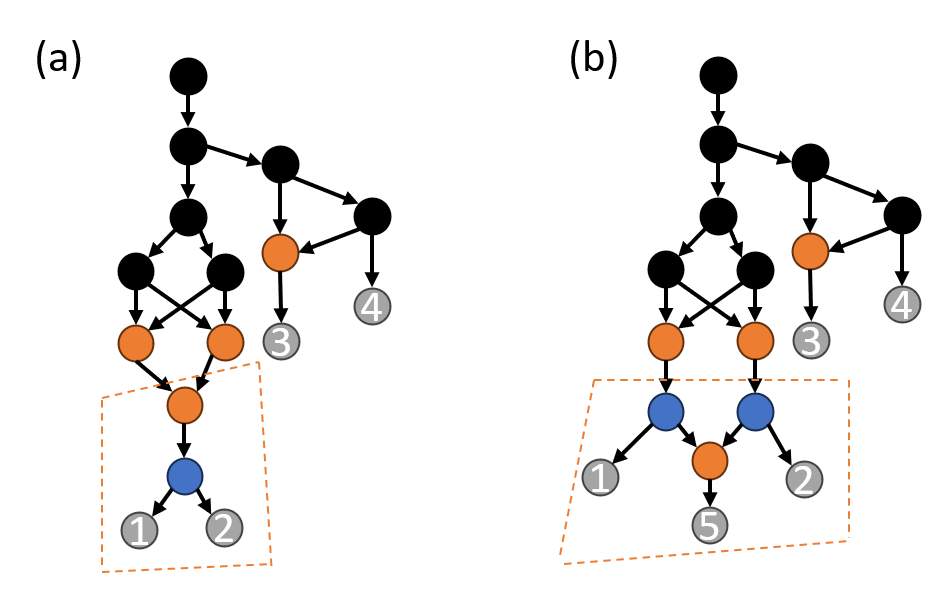}
\caption{\label{4_ret_fig_4} {\bf Symmetric and asymmetric topologies of the bottom subnetwork in a network with the component graph on the right-hand side in the third group.} The circled structures denote the bottom subnetwork. (a) Symmetric bottom topology: ${\cal T}_2$ contains only $\ell$, which is inserted into the edge from the root of ${\cal T}_3$. (b) Asymmetric bottom topology. }
\end{figure}

    For the asymmetric case, we discuss two scenarios. If ${\cal T}_2$ has no network leaves, which means ${\cal T}_3$ has $n-j-i$ network leaves, then we need to insert $\ell$ into any edge except the edge from the root. This provides $(2n-2j-2i-2)t_{n-j-i}$ possibilities. If both ${\cal T}_2$ and ${\cal T}_3$ contain at least one network leaf, since ${\cal T}_2$ and ${\cal T}_3$ should be indistinguishable, we have $$\frac{1}{2}\displaystyle\sum_{m=1}^{n-j-i-1}\binom{n-j-i}{m}t_{m+1}t_{n-j-i-m}(2n-2j-2i-2m-2)$$ possibilities. Consequently, we have
    
    \begin{eqnarray*}
        A_{low}(n-j-i)&=&(2n-2j-2i-2)t_{n-j-i}\\&&+\sum_{m=1}^{n-j-i-1}\binom{n-j-i}{m}t_{m+1}t_{n-j-i-m}(n-j-i-m-1)\\
        &=&(2n-2j-2i-2)\frac{(2n-2j-2i-2)!}{2^{n-j-i-1}(n-j-i-1)!}+2^{n-j-i-1}(n-j-i)!
        \\&&-\frac{(2n-2j-2i-2)!}{(n-j-i-1)!2^{n-j-i-1}}(3n-3j-3i-2).
    \end{eqnarray*}

    Consequently, when the top subnetwork contains symmetric topologies, the number of networks (denoted by $P_2$) is given by
    {\small
\begin{align*}
P_2
&=
\begin{aligned}[t]
&\sum_{i=1}^{n-1}\binom{n}{i}t_i
 \sum_{j=0}^{n-i-1}\binom{n-i}{j}S_{\mathrm{top}}(j)\times
 \left[
 S_{\mathrm{low}}(n-j-i)+A_{\mathrm{low}}(n-j-i)
 \right]
\end{aligned}
\\
&=
\begin{aligned}[t]
&\sum_{i=1}^{n-1}\binom{n}{i}t_i
 \sum_{j=0}^{n-i-1}\binom{n-i}{j}
 \frac{(2j)!}{2^j j!}(2j^3+9j^2+13j+7) \\
&\qquad {}\times
 \Bigg[
 \frac{(2n-2j-2i-1)!}
 {2^{n-j-i-1}(n-j-i-1)!}
 +2^{n-j-i-1}(n-j-i)! \\
&\hspace{1cm}
 {}
 -\frac{(2n-2j-2i-2)!}
 {(n-j-i-1)!2^{n-j-i-1}}
 (3n-3j-3i-2)
 \Bigg]
\end{aligned}
\\
&=
\begin{aligned}[t]
&\sum_{i=1}^{n-1}\binom{n}{i}t_i
 \Bigg\{
 (n-i)!2^{n-i-1}
 \sum_{j=0}^{n-i-1}\binom{2j}{j}
 \frac{2j^3+9j^2+13j+7}{2^{2j}} \\
&\qquad {}
 -\frac{(n-i)!}{2^{n-i}}
 \sum_{j=0}^{n-i-1}
 \binom{2j}{j}\binom{2n-2i-2j}{n-i-j} \\
&\hspace{1cm}
 {}\times
 \frac{(2j^3+9j^2+13j+7)(n-i-j-1)}
 {2n-2i-2j-1}
 \Bigg\}
\end{aligned}
\\
&=
\begin{aligned}[t]
&\sum_{i=1}^{n-1}
 \frac{n!}{(n-i)!i!}
 \frac{(2i-2)!}{2^{i-1}(i-1)!}
 \Bigg\{
 \frac{(2n-2i-1)!}
 {7(n-i-1)!2^{n-i-1}} \\
&\qquad {}\times
 \left[
 2(n-i)^4+23(n-i)^3+78(n-i)^2
 +114(n-i)+49
 \right] \\
&\qquad {}
 -2^{n-i-4}(n-i)!
 \left[
 5(n-i)^3+45(n-i)^2+106(n-i)+92
 \right]
 \Bigg\}
\end{aligned}
\\
&=
\begin{aligned}[t]
&\frac{n!}{7\cdot2^n}
 \sum_{i=1}^{n-1}
 \binom{2i}{i}\binom{2n-2i}{n-i}
 \frac{1}{2i-1} \\
&\qquad {}\times
 \left[
 2(n-i)^4+23(n-i)^3+78(n-i)^2
 +114(n-i)+49
 \right] \\
&\quad {}
 -n!2^{n-4}
 \sum_{i=1}^{n-1}
 \binom{2i}{i}
 \frac{
 5(n-i)^3+45(n-i)^2+106(n-i)+92
 }{2^{2i}(2i-1)}
\end{aligned}
\\
&=
\begin{aligned}[t]
&\frac{(2n-3)!}{7\cdot2^{n-3}(n-2)!}
 \left(
 2n^5+29n^4+112n^3+131n^2-8n-49
 \right) \\
&\quad {}
 -2^{n-3}n!
 \left(
 5n^3+45n^2+106n+92
 \right).
\end{aligned}
\end{align*}
}

    When the top subnetwork contains no symmetric topologies, the number of networks (denoted by $P_3$) is given by
   {\small
\begin{align*}
P_3
&=
\begin{aligned}[t]
&\sum_{i=1}^{n-1}\binom{n}{i}t_i
 \sum_{j=0}^{n-i-1}\binom{n-i}{j}A_{\mathrm{top}}(j) \times
 \left[
 \frac{1}{2}S_{\mathrm{low}}(n-j-i)
 +A_{\mathrm{low}}(n-j-i)
 \right]
\end{aligned}
\\
&=
\begin{aligned}[t]
&\sum_{i=1}^{n-1}\binom{n}{i}t_i
 \sum_{j=0}^{n-i-1}\binom{n-i}{j}
 \Bigg[
 \frac{(j+3)^3(2j+4)!}{2^{j+2}(j+2)!}
 +\frac{(j+3)(j+2)(2j+2)!}{2^{j+1}(j+1)!} \\
&\hspace{1cm}
 {}
 -(2j^3+9j^2+16j+7)\frac{(2j)!}{2^j j!}
 \Bigg] \\
&\qquad {}\times
 \Bigg[
 2^{n-j-i-1}(n-j-i)!
 -\frac{(2n-2j-2i-2)!}
 {(n-j-i-1)!2^{n-j-i}}
 (2n-2j-2i-1)
 \Bigg]
\end{aligned}
\\
&=
\begin{aligned}[t]
&\sum_{i=1}^{n-1}\binom{n}{i}t_i
 \sum_{j=0}^{n-i-1}
 \frac{(n-i)!}{j!(n-i-j)!}
 \frac{(2j)!}{j!2^j} \\
&\qquad {}\times
 \left(
 4j^5+44j^4+183j^3+353j^2+298j+80
 \right) \\
&\qquad {}\times
 \Bigg[
 2^{n-j-i-1}(n-j-i)!
 -\frac{(2n-2j-2i-2)!}
 {(n-j-i-1)!2^{n-j-i}}
 (2n-2j-2i-1)
 \Bigg]
\end{aligned}
\\
&=
\begin{aligned}[t]
&\sum_{i=1}^{n-1}\binom{n}{i}t_i
 \Bigg\{
 2^{n-i-1}(n-i)!
 \sum_{j=0}^{n-i-1}\binom{2j}{j}
 \frac{
 4j^5+44j^4+183j^3+353j^2+298j+80
 }{2^{2j}} \\
&\qquad {}
 -\frac{(n-i)!}{2^{n-i+1}}
 \sum_{j=0}^{n-i-1}
 \binom{2j}{j}\binom{2n-2i-2j}{n-i-j} \\
&\hspace{1cm}
 {}\times
 \left(
 4j^5+44j^4+183j^3+353j^2+298j+80
 \right)
 \Bigg\}
\end{aligned}
\\
&=
\begin{aligned}[t]
&\sum_{i=1}^{n-1}
 \frac{n!}{(n-i)!i!}
 \frac{(2i-2)!}{2^{i-1}(i-1)!}
 \Bigg\{
 \frac{(2n-2i-1)!}
 {231(n-i-1)!2^{n-i}} \\
&\qquad {}\times
 \Bigl[
 168(n-i)^6+2716(n-i)^5+17426(n-i)^4
 +56129(n-i)^3 \\
&\hspace{1.2cm}
 {}
 +92881(n-i)^2+71382(n-i)+18480
 \Bigr] \\
&\qquad {}
 -2^{n-i-7}(n-i)!
 \Bigl[
 63(n-i)^5+840(n-i)^4+4325(n-i)^3 \\
&\hspace{1.2cm}
 {}
 +10680(n-i)^2+12316(n-i)+5120
 \Bigr]
 \Bigg\}
\end{aligned}
\\
&=
\begin{aligned}[t]
&\frac{n!}{231\cdot2^{n+1}}
 \sum_{i=1}^{n-1}
 \binom{2i}{i}\binom{2n-2i}{n-i}
 \frac{1}{2i-1} \\
&\qquad {}\times
 \Bigl[
 168(n-i)^6+2716(n-i)^5+17426(n-i)^4
 +56129(n-i)^3 \\
&\hspace{1.2cm}
 {}
 +92881(n-i)^2+71382(n-i)+18480
 \Bigr] \\
&\quad {}
 -n!2^{n-7}
 \sum_{i=1}^{n-1}\binom{2i}{i}
 \frac{1}{2^{2i}(2i-1)} \\
&\qquad {}\times
 \Bigl[
 63(n-i)^5+840(n-i)^4+4325(n-i)^3
 +10680(n-i)^2 \\
&\hspace{1.2cm}
 {}
 +12316(n-i)+5120
 \Bigr]
\end{aligned}
\\
&=
\begin{aligned}[t]
&\frac{(2n-2)!}{231(n-1)!2^{n-1}}
 \Bigl(
 168n^7+3556n^6+25770n^5+84607n^4 \\
&\hspace{0.5cm}
 {}
 +125223n^3+53008n^2-33150n-18480
 \Bigr) \\
&\quad {}
 -2^{n-6}n!
 \left(
 63n^5+840n^4+4325n^3
 +10680n^2+12316n+5120
 \right).
\end{aligned}
\end{align*}
}

    Summing the counts $P_1$, $P_2$ and $P_3$ together, we obtain the total number of networks corresponding to the third group of component graphs:
    \begin{equation*}
\resizebox{0.98\linewidth}{!}{$\displaystyle
\begin{aligned}
C_3={}&
n!2^{n-9}
\left(
21n^6+147n^5-2175n^4-21635n^3
-70854n^2-97312n-46272
\right)
\\
&-\frac{(2n-2)!}{231(n-1)!2^{n-1}}
\left(
84n^7-532n^6-13383n^5-64732n^4
-121380n^3-66424n^2+29823n+21714
\right).
\end{aligned}
$}
\end{equation*}
\end{proof}

\begin{lemma}
\label{lem34}
    Let $n\geq 2$. The total number of networks in ${\cal P}_{n,4}$ associated with the fourth group of component graphs in Figure \ref{fig:41} is given by:
    $$C_4=\frac{(2n+3)!}{(n+1)!2^{n+1}}(n+3)^2(2n+7)+2^{n-7}(n+2)!(21n^4-28n^3-1943n^2-8994n-12096).$$
\end{lemma}

\begin{proof}
    A network associated with this group of component graphs can be divided into top and bottom subnetworks. The top subnetwork is a one-component network with two reticulations, two designated leaves $\ell_{1}$ and $\ell_{2}$, and $j$ network leaves. The bottom subnetwork is a forest of two trees with $n-j$ leaves, where $\ell_1$ and $\ell_2$ are inserted into arbitrary tree edges one by one in the forest in order to form two reticulations. The two insertion operations are unordered. Each resulting network is
obtained from exactly \(2!\) ordered insertion sequences, so we divide by
\(2!=2\). Then we have:
{\small
\[
\begin{aligned}
C_4
&=
\frac{1}{2}
\sum_{j=0}^{n-2}
\binom{n}{j}
\lvert{\cal OCP}_{j+4,2}\rvert
\lvert{\cal F}_{n-j,2}\rvert \\
&\qquad {}\times
(2n-2j-2)(2n-2j-1)
\\
&=
\sum_{j=0}^{n-2}
\frac{n!}{(n-j)!j!}
\frac{(j+3)^2(2j+7)(2j+4)!}
     {2^{j+2}(j+2)!} \\
&\qquad {}\times
\frac{(2n-2j-2)!}
     {2^{n-j-1}(n-j-1)!}
(n-j-1)(2n-2j-1)
\\
&=
\frac{(2n+3)!}{(n+1)!2^{n+1}}
(n+3)^2(2n+7) \\
&\qquad {}
+2^{n-7}(n+2)!
\left(
21n^4-28n^3-1943n^2-8994n-12096
\right).
\end{aligned}
\]
}
    
\end{proof}

\begin{lemma}
\label{lem35}
    Let $n\geq 1$. The total number of networks in ${\cal P}_{n,4}$ associated with the fifth group of component graphs in Figure \ref{fig:41} is given by:
    $$C_5=\frac{1}{3}2^{n-6}(n+4)!n(7n+29).$$
\end{lemma}
\begin{proof}
    A network associated with this group of component graphs can be divided into top and bottom subnetworks. The top subnetwork is a one-component network with one reticulation, three designated leaves, and $j\ge0$ network leaves. The bottom subnetwork is a tree with $n-j$ leaves, where the three designated leaves are inserted into arbitrary tree edges one by one to form three reticulations. The three insertion operations are unordered. Each resulting network is
obtained from exactly \(3!\) ordered insertion sequences, so we divide by
\(3!=6\). Then we have:
\begin{eqnarray*}
    C_5&=&\frac{1}{6}\sum_{j=0}^{n-1} \binom{n}{j} |{\cal OCP}_{j+4,1}|\times t_{n-j} \times (2n-2j-1)(2n-2j)(2n-2j+1)\\
    &=&\frac{1}{6}\sum_{j=0}^{n-1} \frac{n!}{(n-j)!j!}\frac{(2j+6)!}{2^{j+3}(j+2)!}\frac{(2n-2j-2)!}{2^{n-j-1}(n-j-1)!} (2n-2j-1)(2n-2j)(2n-2j+1)\\
    &=&\frac{1}{3}2^{n-6}(n+4)!n(7n+29).
\end{eqnarray*}

\end{proof}

\begin{figure}[b!]
    \centering
    \includegraphics[width=0.4\linewidth]{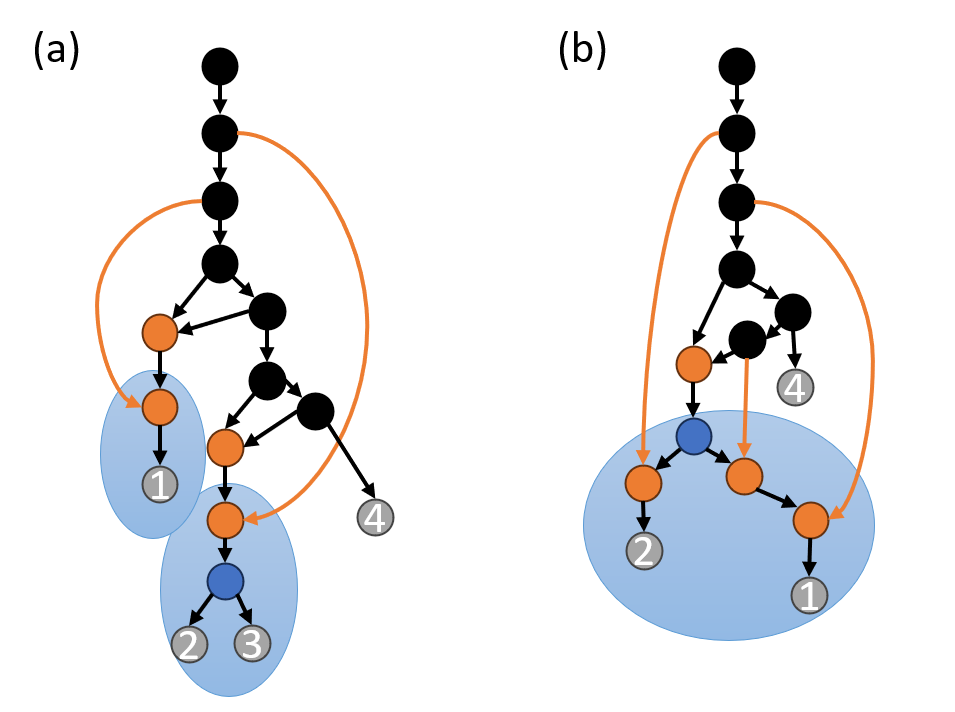}
    \caption{{\bf Illustration of network decomposition in the proofs of Lemmas \ref{lem34} and \ref{lem35}.} The unshaded subnetwork is the top subnetwork, while the shaded subnetwork is the bottom subnetwork $P$ in each case. (a) $P$ is a forest of two trees. (b) $P$ is a tree.
    \label{4_ret_fig_5}}
    \end{figure}

\begin{lemma}\label{lem36}
    Let $n \geq 1$. The total number of networks in ${\cal P}_{n,4}$ associated with the sixth group of component graphs in Figure \ref{fig:41} is given by:
    \begin{eqnarray*}
       C_6&=&2^{n-7}(n+1)!(7n^5+86n^4+605n^3+2462n^2+5000n+3968)\\
        &&-\frac{(2n+1)!}{3465\cdot n!2^{n+1}}(560n^5+7380n^4+52996n^3+194211n^2+333903n+214830).
    \end{eqnarray*}
\end{lemma}

\begin{proof}
    For a network $N$ associated with this group of component graphs, we decompose it into two subnetworks. Its top subnetwork is a one-component network with two reticulations, each having a designated child, and one designated leaf. The bottom subnetwork is a `subnetwork' with one reticulation, where parallel edges are allowed, and two designated roots, each of in-degree 0 and out-degree 1. Then, we need to insert the designated leaf into an edge of the bottom subnetwork. Therefore, the network $N$ can be reconstructed from these two parts by merging the designated children of reticulations in the top subnetwork with the designated roots in the bottom subnetwork, then inserting the designated leaf into some edge in the bottom subnetwork. 

    First, suppose the top subnetwork contains $j$ network leaves. Then it has $|{\cal OCP}_{j+3,2}|$ possible topologies. Among these topologies, some are \textit{symmetric}, as switching the two designated children of the two reticulations results in the same subnetwork. In this case, the reticulations must share the same two parents, denoted by $p_1$, $p_2$. 

    If $p_1$ and $p_2$ share the same parent (see Figure \ref{4_ret_fig_6}a), denoted by $p$, such a symmetric subnetwork can be obtained by inserting $p$ into a tree with $j+1$ leaves (where leaf $j+1$ is designated). Therefore, there are $(2j+1)t_{j+1}$ possibilities.

    If $p_1$ and $p_2$ have distinct parents (see Figure \ref{4_ret_fig_6}b), denoted by $p_1'$ and $p_2'$, such a symmetric subnetwork can be obtained by inserting $p_1'$ and $p_2'$ into either the same edge or two distinct edges of a tree with $j+1$ leaves. Thus, there are $\left[(2j+1)+\binom{2j+1}{2}\right]t_{j+1}$ possibilities.

    These two mutually exclusive cases exhaust all symmetric topologies,
since \(p_1\) and \(p_2\) either have a common parent or have distinct
parents.

    \begin{figure}[b!]
    \centering
    \includegraphics[width=0.5\linewidth]{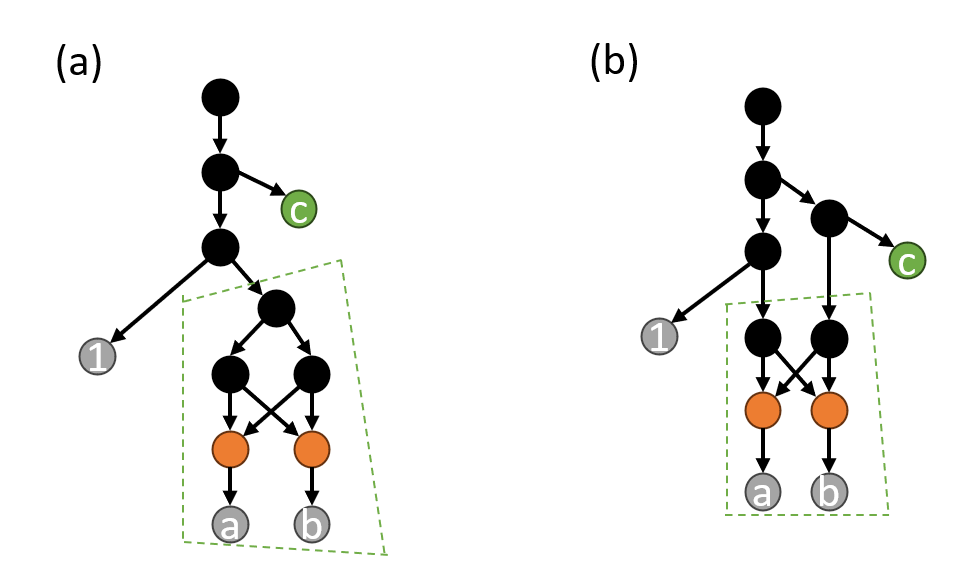}
    \caption{{\bf Symmetric topologies of the top subnetwork in a network with a component graph in the sixth group.} The green leaf $c$ corresponds to the designated leaf in the top subnetwork. The circled part corresponds to the relevant symmetric topology. (a) First symmetric topology: the two common parents of the reticulations are siblings. (b) Second symmetric topology: the two common parents of the reticulations are not siblings. 
    \label{4_ret_fig_6}}
    \end{figure}

    Consequently, for $j\geq0$, the number of possible top subnetworks containing symmetric topologies is given by
    $$S_{top}(j)=(2j+1)t_{j+1}+\left[(2j+1)+\binom{2j+1}{2}\right]t_{j+1}=\frac{(j+2)(2j+1)!}{2^j j!}.$$
    
    Then, the number of remaining top subnetworks containing no symmetric topologies is given by
    $$A_{top}(j)=|{\cal OCP}_{j+3,2}|-S_{top}(j)=\frac{(j+2)(2j+1)!}{2^jj!}(2j^2+9j+9).$$

    \begin{figure}[b!]
    \centering
    \includegraphics[width=0.6\linewidth]{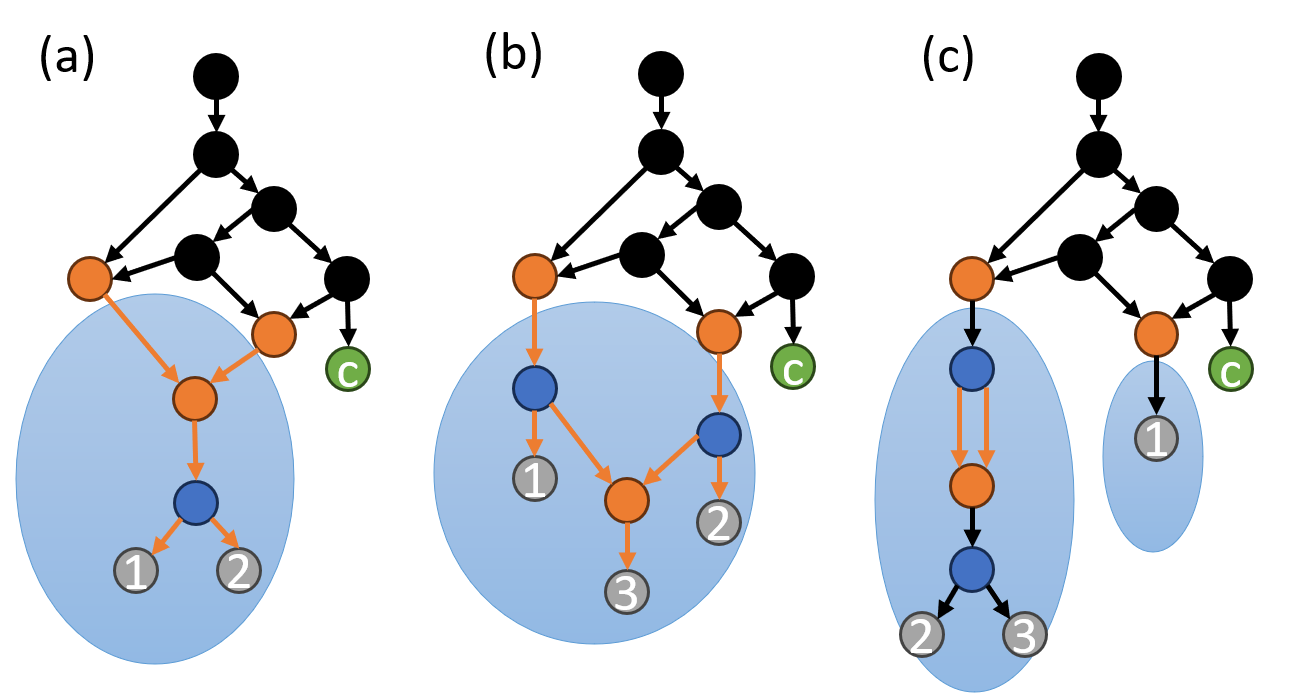}
    \caption{{\bf The possible decomposition of a `network' associated with a component graph in the sixth group \textit{before} the designated leaf from the top subnetwork is inserted.} The unshaded part corresponds to the top subnetwork, while the shaded part corresponds to the bottom subnetwork. The green leaf $c$ corresponds to the designated leaf in the top subnetwork. The orange edges in the bottom subnetwork correspond to the possible edges where leaf $c$ can be inserted. (a) Symmetric bottom structure: the two parents of the reticulation below are just the two reticulations above. (b) First asymmetric topology: the two parents of the reticulation below are not the two reticulations above. (c) Second asymmetric topology: the bottom subnetwork is a `network' with a pair of parallel edges and two designated roots.
    \label{4_ret_fig_7}}
    \end{figure}

    Now we consider the bottom subnetwork. It is a `network' with one reticulation (denoted by $r$) and two designated roots, where parallel edges are allowed. If the two parents of $r$ are just the two designated roots, then we say the bottom subnetwork is \textit{symmetric}. The part below $r$ is a tree with $n-j$ network leaves, which provides $t_{n-j}$ possibilities. See Figure \ref{4_ret_fig_7} for illustration. The designated leaf can then be inserted into an arbitrary edge of the bottom subnetwork.
    
    Note that if the top subnetwork is symmetric, the two edges entering $r$ should be indistinguishable. This provides $2n-2j$ choices to insert the designated leaf. If the top subnetwork is not symmetric, the two edges entering $r$ will be distinguishable. The number of choices to insert the designated leaf will become $2n-2j+1$. Consequently, when the bottom subnetwork is symmetric, we discuss the following two scenarios.

    Suppose both the top subnetwork and bottom subnetwork are symmetric. Let $P_1$ denote the number of such networks. Then
    \begin{eqnarray*}
        P_1&=&\sum_{j=0}^{n-1} \binom{n}{j} S_{top}(j)\times t_{n-j}(2n-2j)\\
        &=&\sum_{j=0}^{n-1} \frac{n!}{(n-j)!j!} \frac{(j+2)(2j+1)!}{2^j j!} \frac{(2n-2j-2)!}{2^{n-j-1}(n-j-1)!}(2n-2j)\\
        &=&\frac{n!}{2^{n-1}}\sum_{j=0}^{n-1} \binom{2j}{j} \binom{2n-2j}{n-j} \frac{(2j+1)(j+2)(n-j)}{2n-2j-1}\\
        &=& 2^{n-2}n!n(3n+5).
     \end{eqnarray*}

     Suppose the top subnetwork is not symmetric while the bottom subnetwork is symmetric. Let $P_2$ denote the number of such networks. Then
    \begin{eqnarray*}
        P_2&=& \frac{1}{2} \sum_{j=0}^{n-1} \binom{n}{j} A_{top}(j)\times t_{n-j}(2n-2j+1)\\
        &=&\frac{1}{2} \sum_{j=0}^{n-1}  \frac{n!}{(n-j)!j!} \frac{(j+2)(2j+1)!}{2^jj!}(2j^2+9j+9) \frac{(2n-2j-2)!}{2^{n-j-1}(n-j-1)!}(2n-2j+1)\\
        &=& \frac{n!}{2^{n+1}}\sum_{j=0}^{n-1} \binom{2j}{j} \binom{2n-2j}{n-j} \frac{(j+2)(2j+1)(2n-2j+1)(2j^2+9j+9)}{2n-2j-1}\\
        &=&\frac{(2n+3)!}{(n+1)!2^{n+2}}(n+3)(n+2)+2^{n-6}(n+1)!(35n^3-5n^2-462n-576).
    \end{eqnarray*}

    Then consider the case when the bottom subnetwork is asymmetric. Recall that it is a `network' with one reticulation and two designated roots, either without (Figure \ref{4_ret_fig_7}b) or with (Figure \ref{4_ret_fig_7}c) parallel edges. For the former, the number of bottom subnetworks is given by $|{\cal P}_{n-j,1}|$ and there are $2n-2j+1$ insertion edges. For the latter, we can first consider the subnetwork as a forest of two trees. Then replace one edge of the forest with the parallel-edge structure, and then insert the designated leaf into one of the parallel edges. This provides $|{\cal F}_{n-j,2}|(2n-2j-2)$ choices.

    Now suppose the bottom subnetwork is asymmetric. Let $P_3$ denote the number of such networks. Then
    
    {\small
\begin{align*}
P_3
&=
\begin{aligned}[t]
&\sum_{j=0}^{n-2}
\binom{n}{j}
\left[
S_{\mathrm{top}}(j)+A_{\mathrm{top}}(j)
\right] \times
\Bigl[
\lvert{\cal P}_{n-j,1}\rvert(2n-2j+1)
+\lvert{\cal F}_{n-j,2}\rvert(2n-2j-2)
\Bigr]
\end{aligned}
\\
&=
\begin{aligned}[t]
&\sum_{j=0}^{n-2}
\frac{n!}{(n-j)!j!}
\frac{(j+2)(2j+1)!}{2^j j!}
(2j^2+9j+10) \\
&\qquad {}\times
\Bigg[
\frac{(n-j)(2n-2j+1)!}
     {2^{n-j}(n-j)!}
-2^{n-j-1}(n-j)!(2n-2j+1) \\
&\hspace{1cm}
 {}
+\frac{(2n-2j-2)!}
      {2^{n-j-1}(n-j-1)!}
 (2n-2j-2)
\Bigg]
\end{aligned}
\\
&=
\begin{aligned}[t]
&\frac{n!}{2^n}
\sum_{j=0}^{n-2}
\binom{2j}{j}
\binom{2n-2j}{n-j} \times
(2j+1)(j+2)(2j^2+9j+10)
(2n-2j+1)(n-j)
\\
&\quad {}
+\frac{n!}{2^{n-1}}
\sum_{j=0}^{n-2}
\binom{2j}{j}
\binom{2n-2j}{n-j} \times
\frac{
(2j+1)(j+2)(2j^2+9j+10)(n-j-1)
}{
2n-2j-1
}
\\
&\quad {}
-n!2^{n-1}
\sum_{j=0}^{n-2}
\binom{2j}{j}
\frac{
(2j+1)(j+2)(2j^2+9j+10)(2n-2j+1)
}{
2^{2j}
}
\end{aligned}
\\
&=
\begin{aligned}[t]
&2^{n-7}n!
\Bigl(
7n^6+93n^5+621n^4+3007n^3
+8300n^2+10884n+5120
\Bigr)
\\
&\quad {}
-\frac{(2n+1)!}{3465\,n!2^n}
\Bigl(
280n^5+3690n^4+29963n^3 
+119628n^2+213729n+138600
\Bigr).
\end{aligned}
\end{align*}
}
    
    Adding $P_1$, $P_2$ and $P_3$ together, we have
    \begin{eqnarray*}
        C_6&=& 2^{n-7}(n+1)!(7n^5+86n^4+605n^3+2462n^2+5000n+3968)\\
        &&-\frac{(2n+1)!}{3465\cdot n!2^{n+1}}(560n^5+7380n^4+52996n^3+194211n^2+333903n+214830).
    \end{eqnarray*}

\end{proof}

\begin{lemma}\label{lem37}
    Let $n \geq 1$. The total number of networks in ${\cal P}_{n,4}$ associated with the seventh group of component graphs in Figure \ref{fig:41} is given by:
    \begin{equation*}
\resizebox{0.98\linewidth}{!}{$\displaystyle
\begin{aligned}
C_7={}&
\frac{1}{3}2^{n-9}n!
\Bigl(
15n^6+249n^5+1875n^4+7655n^3
+19230n^2+24160n+10560
\Bigr)
\\
&-\frac{(2n-2)!}{3465(n-1)!2^n}
\Bigl(
640n^7+8912n^6+52204n^5+172946n^4
+293824n^3+151979n^2-77595n-48510
\Bigr).
\end{aligned}
$}
\end{equation*}
\end{lemma}

\begin{proof}
    A network $N$ associated with this group of component graphs can be decomposed into two subnetworks. Its top subnetwork is a one-component network with two reticulations, each of which has a designated child. The bottom subnetwork is a `network' with two reticulations together with two designated roots.

    First consider the top subnetwork. If it contains $j$ network leaves, it will have $|{\cal OCP}_{j+2,2}|$ possible topologies. Among these topologies, some are \textit{symmetric}, as switching the two designated children of the two reticulations results in the same subnetwork. In this case, the reticulations must share the same two parents, denoted by $p_1$, $p_2$. 

    If $p_1$ and $p_2$ share the same parent (see Figure \ref{4_ret_fig_18}a), denoted by $p$, such a symmetric subnetwork can be obtained by inserting $p$ into a tree with $j$ leaves. Therefore, for $j \geq 1$, there are $(2j-1)t_{j}$ possibilities. Also, when $j=0$, the structure can directly form the top subnetwork.

    If $p_1$ and $p_2$ have distinct parents (see Figure \ref{4_ret_fig_18}b), denoted by $p_1'$ and $p_2'$, such a symmetric subnetwork can be obtained by inserting $p_1'$ and $p_2'$ into the same edge or different edges of a tree with $j$ leaves. Thus, for $j\geq 1$, there are $\left[(2j-1)+\binom{2j-1}{2}\right]t_{j}$ possibilities. As in the proof of Lemma~\ref{lem36}, these two cases are mutually exclusive and exhaustive.
    
    Consequently, if the top subnetwork contains $j\ge0$ network leaves, the number of top subnetworks containing symmetric topologies is
    $$S_{top}(j)=\frac{(j+1)(2j)!}{2^jj!},
    $$
    and the number of top subnetworks containing asymmetric topologies is
    $$A_{top}(j)=\frac{(j+1)(2j^2+5j+2)(2j)!}{2^jj!}. $$
    
\begin{figure}[b!]
    \centering
    \includegraphics[width=0.7\linewidth]{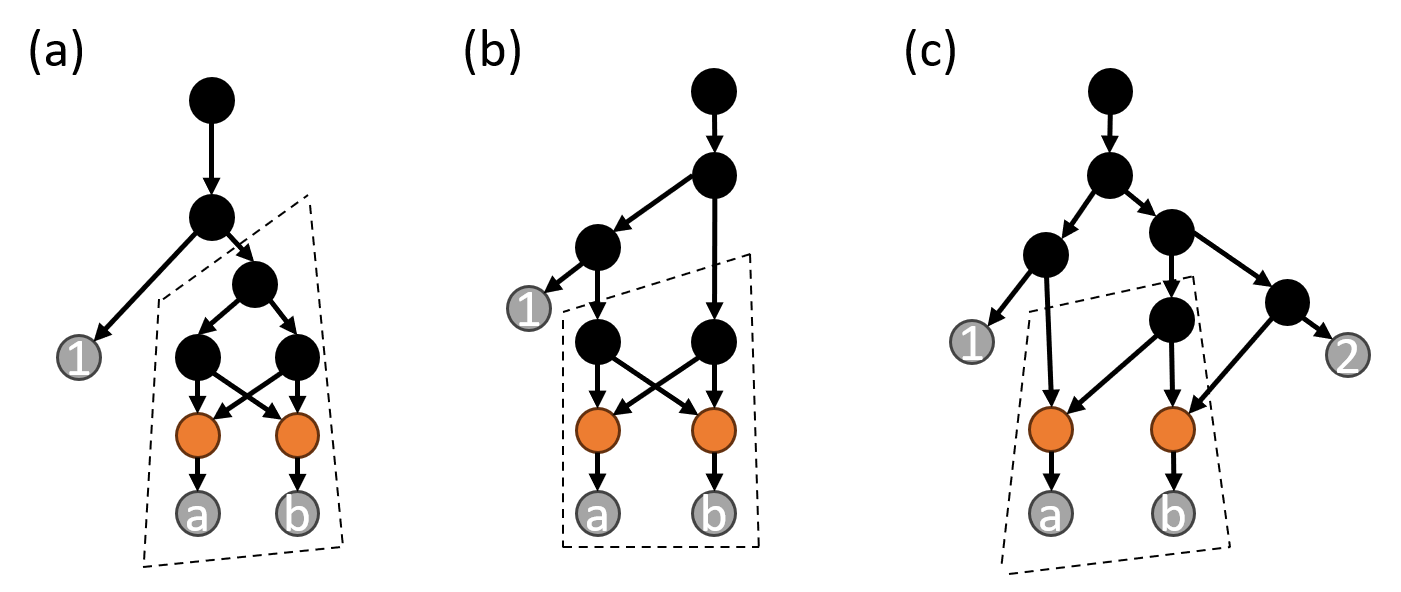}
    \caption{{\bf Symmetric and asymmetric topologies of the top subnetwork in a network with a component graph in the seventh group.} (a) Symmetric topology: the common parents of reticulations are siblings. (b) Symmetric topology: the common parents of reticulations are not siblings. (c) Asymmetric topology.
    \label{4_ret_fig_18}}
    \end{figure}

    Then we discuss the bottom subnetwork. Recall that it is a `network' with two reticulations and two designated roots that are of in-degree 0 and out-degree 1. Among all these topologies, some are \textit{symmetric}. If the bottom subnetwork is symmetric under interchanging its two designated roots, then the two possible assignments of the designated children of the top subnetwork to these roots yield the same network, even when the top subnetwork itself is asymmetric. There are two types of symmetric bottom subnetworks. These two cases are mutually exclusive and exhaust all bottom subnetworks that are invariant under exchanging the two designated roots; they are displayed in
Figure~\ref{4_ret_fig_8}.

    The first type occurs when, after identifying the two designated children of the top subnetwork with the two designated roots of the bottom subnetwork, the two reticulations in the top subnetwork become the common parents of a reticulation \(r\) in the bottom subnetwork (Figure \ref{4_ret_fig_8}a). Therefore, if $n-j\geq 2$, we have

    $$S_{low,1}(n-j)=|{\cal P}_{n-j,1}|. $$

    For the second type, after the same identification, let $p_1$ and $p_2$
denote the two children of the reticulations in the top subnetwork.
The vertices $p_1$ and $p_2$ are the common parents of the two
reticulations $r$ and $s$ in the bottom subnetwork (Figure \ref{4_ret_fig_8}b).
To count this type of symmetric bottom subnetwork, we can regard the
parts below $r$ and $s$ as a forest of two trees. Therefore, if
$n-j\geq 2$, we have
\[
S_{\mathrm{low},2}(n-j)=|\mathcal{F}_{n-j,2}|.
\]
    
    Consequently, combining these two cases, when $n-j\geq 2$, the total number of bottom subnetworks that are symmetric is
    $$S_{low}(n-j)=|{\cal P}_{n-j,1}|+|{\cal F}_{n-j,2}|=\frac{(2n-2j-2)!}{2^{n-j-1}(n-j-1)!}[2(n-j)^2-(n-j)+1]-2^{n-j-1}(n-j)!. $$

    Then, when $n-j\ge 2$, the number of bottom subnetworks that are not symmetric is
    {\footnotesize
    \begin{eqnarray*}
        A_{low}(n-j)&=&|{\cal P}_{n-j,2}|-|{\cal F}_{n-j,2}|\\
        &=&\frac{(2n-2j-2)!}{3\cdot2^{n-j-1}(n-j-1)!}\Big[6(n-j)^4+19(n-j)^3+18(n-j)^2-4(n-j)-6\Big]\\
        &&-2^{n-j-1}(n-j+1)!(2n-2j+3)-\frac{(2n-2j-2)!}{2^{n-j-1}(n-j-1)!}.
    \end{eqnarray*}}

    When $n-j=1$, all possible bottom subnetworks are asymmetric. Hence,
    $$A_{low}(1)=|{\cal P}_{1,2}|=1. $$
    \begin{figure}[b!]
    \centering
    \includegraphics[width=0.6\linewidth]{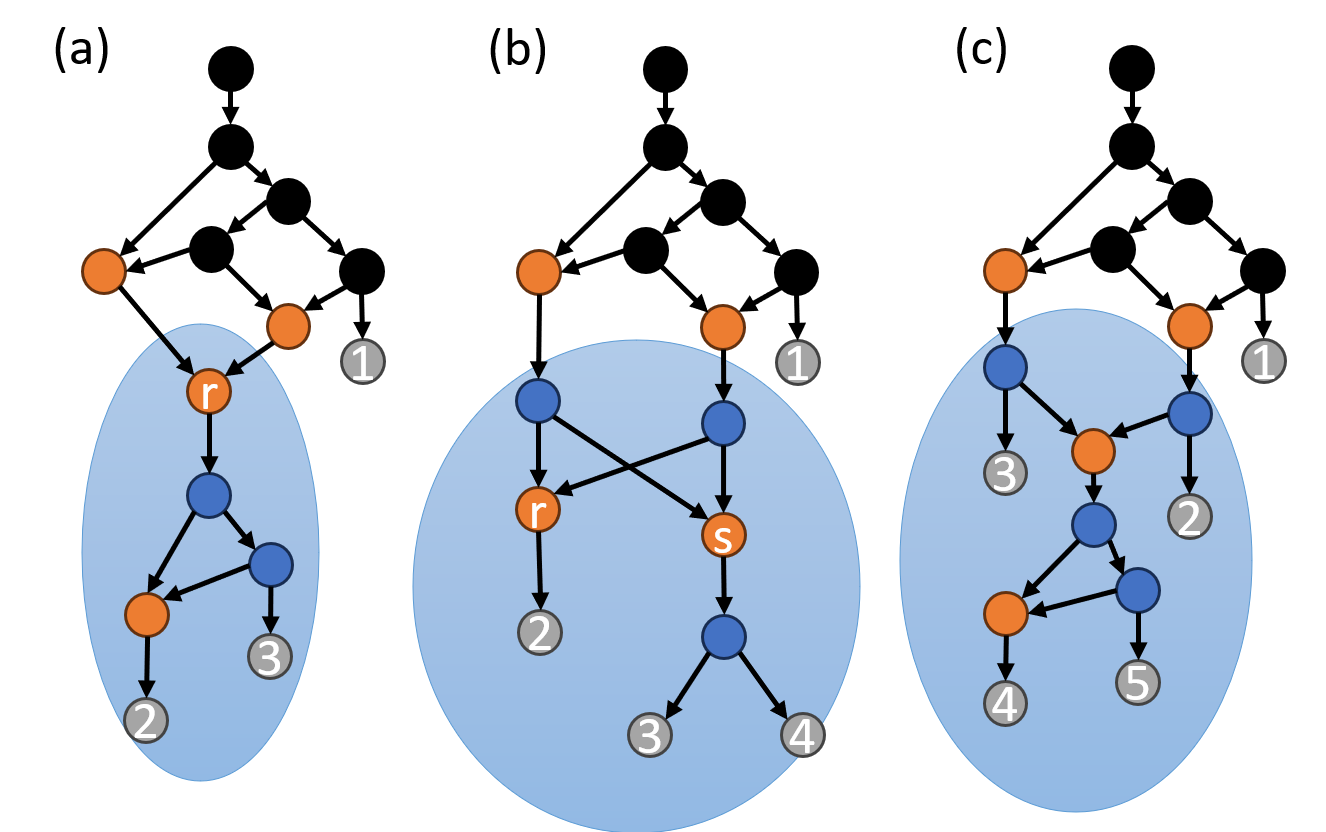}
    \caption{{\bf Symmetric and asymmetric topologies of the bottom subnetwork
in a network with a component graph in the seventh group.}
The shaded part is the bottom subnetwork.
(a) First symmetric topology: the two reticulations in the top subnetwork are the common parents of a reticulation \(r\) in the bottom subnetwork.
(b) Second symmetric topology: the two designated roots have children
\(p_1\) and \(p_2\), which are the common parents of the reticulations
\(r\) and \(s\).
(c) An asymmetric topology.
    \label{4_ret_fig_8}}
    \end{figure}

    Now suppose the bottom subnetwork is symmetric and contains at least 2 network leaves. Let $P_1$ denote the number of such networks. Then
    {\small
\begin{align*}
P_1
&=
\begin{aligned}[t]
&\sum_{j=0}^{n-2}
\binom{n}{j}
\left[
S_{\mathrm{top}}(j)+\frac{1}{2}A_{\mathrm{top}}(j)
\right]
S_{\mathrm{low}}(n-j)
\end{aligned}
\\
&=
\begin{aligned}[t]
&\sum_{j=0}^{n-2}
\frac{n!}{(n-j)!j!}
\frac{(j+1)(2j^2+5j+4)(2j)!}
     {2^{j+1}j!}
\\
&\qquad {}\times
\Bigg\{
\frac{(2n-2j-2)!}
     {2^{n-j-1}(n-j-1)!}
\left[
2(n-j)^2-(n-j)+1
\right] \\
&\hspace{1cm}
{}
-2^{n-j-1}(n-j)!
\Bigg\}
\end{aligned}
\\
&=
\begin{aligned}[t]
&\frac{n!}{2^{n+1}}
\sum_{j=0}^{n-2}
\binom{2j}{j}
\binom{2n-2j}{n-j}
\frac{(j+1)(2j^2+5j+4)}
     {2n-2j-1}
\\
&\qquad {}\times
\left[
2(n-j)^2-(n-j)+1
\right]
\\
&\quad {}
-n!2^{n-2}
\sum_{j=0}^{n-2}
\binom{2j}{j}
\frac{(j+1)(2j^2+5j+4)}
     {2^{2j}}
\end{aligned}
\\
&=
\begin{aligned}[t]
&2^{n-7}n!(n+6)
\left(
5n^3-8n^2-29n-32
\right)
\\
&\quad {}
-\frac{(2n-2)!}
      {35(n-1)!2^n}
\left(
20n^5-18n^4-322n^3-293n^2-17n+140
\right).
\end{aligned}
\end{align*}
}

    Suppose the bottom subnetwork is asymmetric and contains at least 2 network leaves. Let $P_2$ denote the number of such networks. Then
    {\small
\begin{align*}
P_2
&=
\begin{aligned}[t]
&\sum_{j=0}^{n-2}
\binom{n}{j}
\left[
S_{\mathrm{top}}(j)+A_{\mathrm{top}}(j)
\right]
A_{\mathrm{low}}(n-j)
\end{aligned}
\\
&=
\begin{aligned}[t]
&\sum_{j=0}^{n-2}
\frac{n!}{(n-j)!j!}
\frac{(j+1)^2(2j+3)(2j)!}
     {2^j j!}
\\
&\qquad {}\times
\Bigg\{
\frac{(2n-2j-2)!}
     {3\cdot2^{n-j-1}(n-j-1)!}
\Bigl[
6(n-j)^4+19(n-j)^3+18(n-j)^2 \\
&\hspace{1cm}
{}
-4(n-j)-6
\Bigr]
\\
&\hspace{1cm}
{}
-2^{n-j-1}(n-j+1)!(2n-2j+3)
\\
&\hspace{1cm}
{}
-\frac{(2n-2j-2)!}
      {2^{n-j-1}(n-j-1)!}
\Bigg\}
\end{aligned}
\\
&=
\begin{aligned}[t]
&\frac{n!}{3\cdot2^n}
\sum_{j=0}^{n-2}
\binom{2j}{j}
\binom{2n-2j}{n-j}
\frac{(j+1)^2(2j+3)}
     {2n-2j-1}
\\
&\qquad {}\times
\Bigl[
6(n-j)^4+19(n-j)^3+18(n-j)^2
-4(n-j)-6
\Bigr]
\\
&\quad {}
-n!2^{n-1}
\sum_{j=0}^{n-2}
\binom{2j}{j}
\frac{
(j+1)^2(2j+3)(n-j+1)(2n-2j+3)
}{
2^{2j}
}
\\
&\quad {}
-\frac{n!}{2^n}
\sum_{j=0}^{n-2}
\binom{2j}{j}
\binom{2n-2j}{n-j}
\frac{(2j+3)(j+1)^2}
     {2n-2j-1}
\end{aligned}
\\
&=
\begin{aligned}[t]
&\frac{1}{3}2^{n-9}n!
\Bigl(
15n^6+249n^5+1815n^4+7391n^3 
+20154n^2+26632n+12864
\Bigr)
\\
&\quad {}
-\frac{(2n-1)!}
      {3465(n-1)!2^{n-1}}
(2n+3)
\Bigl(
80n^5+1034n^4+5304n^3\\
&\hspace{0.5cm}
{}
+19045n^2+26512n+10395
\Bigr).
\end{aligned}
\end{align*}
}

    If the bottom subnetwork contains only one network leaf, then the top subnetwork contains $n-1$ network leaves. Let $P_3$ denote the number of such networks. Then
    \begin{eqnarray*}
        P_3&=&\binom{n}{n-1}\left[S_{top}(n-1)+A_{top}(n-1) \right]A_{low}(1)   \\
        &=& \frac{n^3(2n+1)(2n-2)!}{2^{n-1}(n-1)!}.
    \end{eqnarray*}

By adding $P_1$, $P_2$ and $P_3$ together, we get
\begin{equation*}
\resizebox{0.98\linewidth}{!}{$\displaystyle
\begin{aligned}
C_7={}&
\frac{1}{3}2^{n-9}n!
\Bigl(
15n^6+249n^5+1875n^4+7655n^3
+19230n^2+24160n+10560
\Bigr)
\\
&-\frac{(2n-2)!}{3465(n-1)!2^n}
\Bigl(
640n^7+8912n^6+52204n^5+172946n^4
+293824n^3+151979n^2-77595n-48510
\Bigr).
\end{aligned}
$}
\end{equation*}
    
\end{proof}

\begin{lemma}\label{lem38}
    Let $n \geq 1$. The total number of networks in ${\cal P}_{n,4}$ associated with the eighth group of component graphs in Figure \ref{fig:42} is given by:
\begin{equation*}
\resizebox{0.98\linewidth}{!}{$\displaystyle
\begin{aligned}
C_8={}&
2^{n-7}n!
\Bigl(
5n^6+75n^5+479n^4+1641n^3
+3164n^2+3148n+1280
\Bigr)
\\
&-\frac{(2n-1)!}{3465(n-1)!2^{n-1}}
\Bigl(
320n^6+4968n^5+30896n^4+96828n^3
+160709n^2+125544n+34650
\Bigr).
\end{aligned}
$}
\end{equation*}
\end{lemma}
\begin{proof}
    A network $N$ associated with this group of component graphs can be decomposed into top and bottom subnetworks. The top subnetwork is a one-component network with one reticulation and two designated leaves $\ell_1$ and $\ell_2$. The bottom subnetwork is a `network' with one reticulation, where parallel edges are allowed. Then $\ell_1$ and $\ell_2$ are inserted into the same or different edges in the bottom subnetwork to form two new reticulations.

    \begin{figure}[b!]
    \centering
    \includegraphics[width=0.6\linewidth]{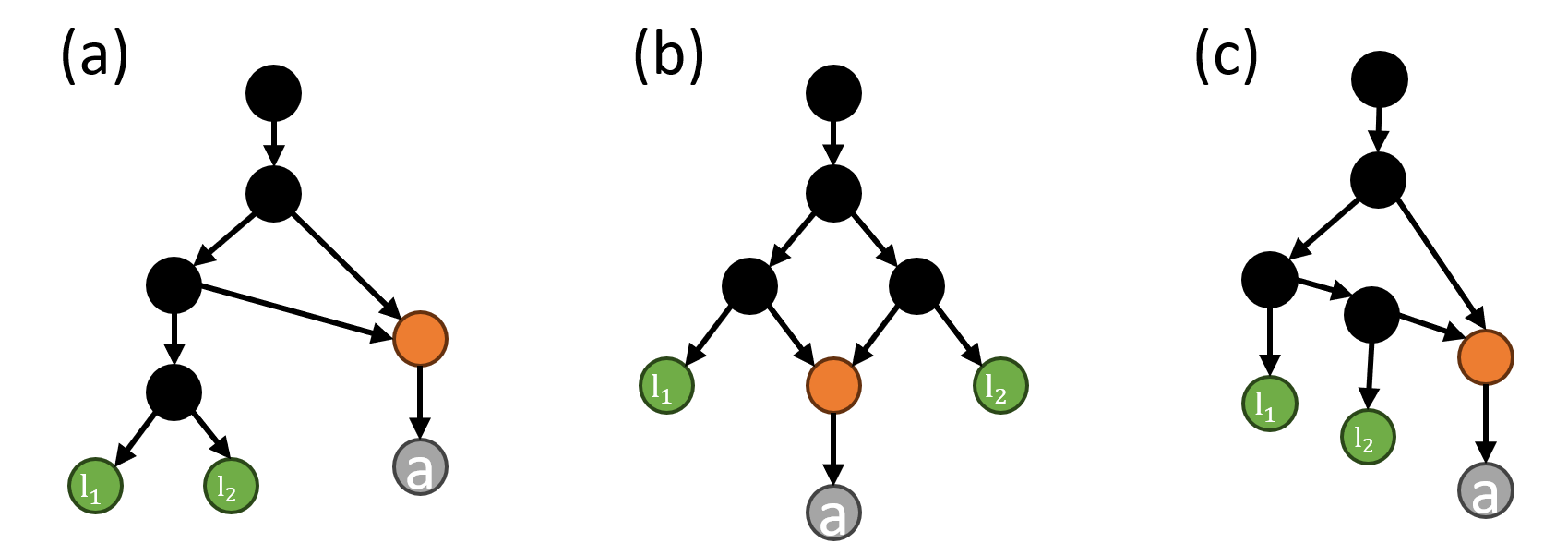}
    \caption{{\bf Symmetric and asymmetric topologies of the top subnetwork in a network with a component graph in the eighth group.} The green leaves $\ell_1$ and $\ell_2$ are two designated leaves. (a) First symmetric topology: $\ell_1$ and $\ell_2$ form a cherry. (b) Second symmetric topology: $\ell_1$ and $\ell_2$ are inserted into the two edges of a parallel-edge pair. (c) Asymmetric topology.
    \label{4_ret_fig_9}}
    \end{figure}

    First consider the top subnetwork. If it contains $j$ network leaves, it will have $|{\cal OCP}_{j+3,1}|$ topologies. Among these topologies, some subnetworks are said to be \textit{symmetric} if switching the two designated leaves $\ell_1$ and $\ell_2$ results in the same subnetwork.

    The first symmetric topology is when $\ell_1$ and $\ell_2$ are siblings (see Figure \ref{4_ret_fig_9}a). It can be obtained by inserting that cherry into an edge from a one-component network (except the edge from the reticulation), which provides $|{\cal OCP}_{j+1,1}|\cdot(2j+3)$ choices. Also, it can be obtained from a `one-component network' that contains a pair of parallel edges by inserting the cherry into one of the parallel edges. This construction yields $t_{j+1}$ possibilities.

    If $\ell_1$ and $\ell_2$ have distinct parents, denote the parents by $p_1$ and $p_2$ respectively. If $p_1$ and $p_2$ have the same parent, and $\ell_1$ and $\ell_2$ have a common sibling (see Figure \ref{4_ret_fig_9}b), then the topology is also symmetric. Such structures can be obtained by inserting $\ell_1$ and $\ell_2$ into the two edges of a pair of parallel edges from a `one-component network' containing one pair of parallel edges. 

    Conversely, every top subnetwork fixed by exchanging
\(\ell_1\) and \(\ell_2\) belongs to one of these two cases. Hence this
enumeration is exhaustive.

    Consequently, for $j\geq 0$, the total number of top subnetworks containing symmetric topologies is given by
    $$S_{top}(j)=|{\cal OCP}_{j+1,1}|\cdot(2j+3)+t_{j+1}+t_{j+1}=\frac{(2j)!}{2^jj!}(2j^2+3j+2).$$

    Thus, the number of asymmetric top subnetworks is
    $$A_{top}(j)=|{\cal OCP}_{j+3,1}|-S_{top}(j)=\frac{(2j)!}{2^{j-1}j!}(2j^3+7j^2+8j+2). $$

    \begin{figure}[b!]
    \centering
    \includegraphics[width=0.6\linewidth]{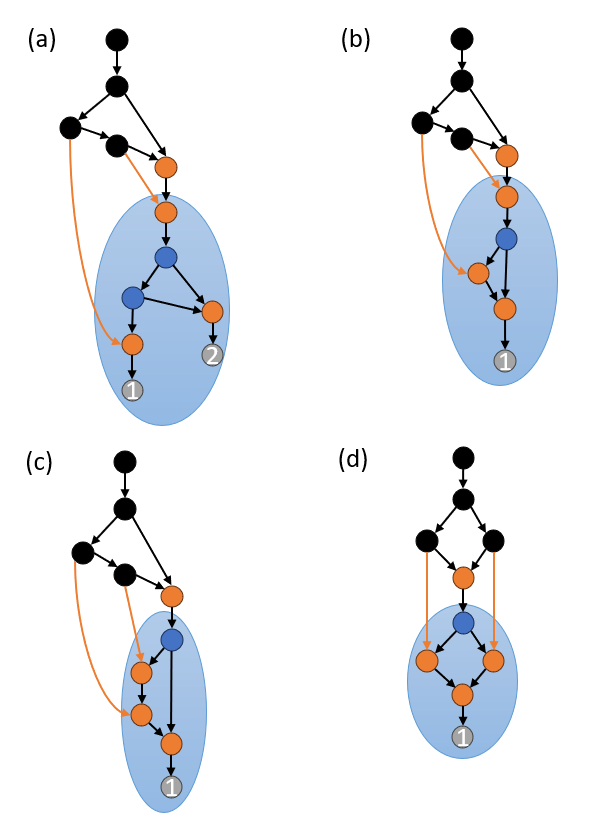}
    \caption{{\bf Possible structures of the bottom subnetwork associated with a component graph in the eighth group.} The orange edges indicate where the designated leaves from the top
subnetwork are inserted into the bottom subnetwork. Shaded parts correspond to the bottom subnetworks. (a) The bottom subnetwork is a network with one reticulation. (b) The bottom subnetwork contains one pair of parallel edges; one designated leaf is inserted into a parallel edge and the other is inserted into a tree edge. (c) The bottom subnetwork contains one pair of parallel edges; both designated leaves are inserted into the same parallel edge. (d) The bottom subnetwork contains one pair of parallel edges; one designated leaf is inserted into a parallel edge and the other designated leaf is inserted into another parallel edge.
    \label{4_ret_fig_10}}
    \end{figure}

    Next, consider the bottom subnetwork. If the bottom subnetwork is a network with one reticulation (Figure \ref{4_ret_fig_10}a), then we just need to insert the designated leaves into edges of the bottom subnetwork one by one and divide by 2 to correct for the two possible insertion orders. Then the number of networks corresponding to this case is

{\small
\begin{align*}
P_1
&=
\begin{aligned}[t]
&\frac{1}{2}
\sum_{j=0}^{n-2}
\binom{n}{j}
\lvert{\cal OCP}_{j+3,1}\rvert
\lvert{\cal P}_{n-j,1}\rvert 
(2n-2j+2)(2n-2j+3)
\end{aligned}
\\
&=
\begin{aligned}[t]
&\frac{1}{2}
\sum_{j=0}^{n-2} 
\frac{n!}{(n-j)!j!}
\frac{(2j+4)!}{2^{j+2}(j+1)!}
\left[
\frac{(n-j)(2n-2j)!}
     {2^{n-j}(n-j)!}
-2^{n-j-1}(n-j)!
\right] \\
&\qquad {}\times
(2n-2j+2)(2n-2j+3)
\end{aligned}
\\
&=
\begin{aligned}[t]
&\frac{n!}{2^n}
\sum_{j=0}^{n-2}
\binom{2j}{j}
\binom{2n-2j}{n-j}
\\
&\qquad {}\times
(2j+1)(2j+3)(j+2)
(n-j)(n-j+1)(2n-2j+3)
\\
&\quad {}
-n!2^{n-1}
\sum_{j=0}^{n-2}
\binom{2j}{j}
\frac{
(2j+1)(2j+3)(j+2)(n-j+1)(2n-2j+3)
}{
2^{2j}
}
\end{aligned}
\\
&=
\begin{aligned}[t]
&2^{n-7}(n+2)!n
\left(
5n^3+60n^2+249n+326
\right)
\\
&\quad {}
-\frac{(2n+3)!}
      {3465(n+1)!2^n}
n\left(
40n^3+541n^2+2750n+3599
\right).
\end{aligned}
\end{align*}
}

For the remaining cases, the bottom subnetwork is a `network' with one reticulation and one pair of parallel edges. There are $t_{n-j}(2n-2j-1)$ possibilities in total.

If only one designated leaf is inserted into the parallel-edge structure (Figure \ref{4_ret_fig_10}b), the other leaf can be inserted in $2n-2j$ ways. Then the number of networks corresponding to this case is
\begin{eqnarray*}
    P_2&=&\sum_{j=0}^{n-1}\binom{n}{j} |{\cal OCP}_{j+3,1}|\cdot t_{n-j} (2n-2j-1)(2n-2j)\\
    &=&\sum_{j=0}^{n-1} \frac{n!}{(n-j)!j!}\frac{(2j+4)!}{2^{j+2}(j+1)!}\frac{(2n-2j-2)!}{2^{n-j-1}(n-j-1)!}(2n-2j-1)(2n-2j)\\
    &=&\frac{n!}{2^{n-1}}\sum_{j=0}^{n-1}\binom{2j}{j}\binom{2n-2j}{n-j}(2j+1)(2j+3)(j+2)(n-j)\\
    &=&2^{n-4}(n+2)!n(5n+11).
\end{eqnarray*}

If both designated leaves are inserted into the parallel-edge structure below, there are still two cases to discuss. First, the two designated leaves are inserted into the same parallel edge (Figure \ref{4_ret_fig_10}c). Then the number of networks corresponding to this case is
\begin{eqnarray*}
    P_3&=&\sum_{j=0}^{n-1}\binom{n}{j} |{\cal OCP}_{j+3,1}|\cdot t_{n-j} (2n-2j-1)\\
    &=&\frac{n!}{2^{n}}\sum_{j=0}^{n-1}\binom{2j}{j}\binom{2n-2j}{n-j}(2j+1)(2j+3)(j+2)\\
    &=&2^{n-2}(n+2)!(5n+12)-\frac{(2n+3)!}{(n+1)!2^{n+1}}(n+2).
\end{eqnarray*}

If the two designated leaves are inserted into different parallel edges (Figure \ref{4_ret_fig_10}d), the scenario will be more complicated since these two parallel edges below are indistinguishable. Using the expressions for the symmetric and asymmetric topologies in the top subnetwork above, the number of networks corresponding to this case is
\begin{eqnarray*}
    P_4&=&\sum_{j=0}^{n-1} \binom{n}{j} [S_{top}(j)+\frac{1}{2}A_{top}(j)]\times t_{n-j}(2n-2j-1)\\
    &=&\sum_{j=0}^{n-1} \frac{n!}{(n-j)!j!} \frac{(2j)!}{2^jj!}(2j^3+9j^2+11j+4) \frac{(2n-2j-1)!}{2^{n-j-1}(n-j-1)!}\\
    &=&\frac{n!}{2^n}\sum_{j=0}^{n-1}\binom{2j}{j}\binom{2n-2j}{n-j}(2j^3+9j^2+11j+4)\\
    &=&2^{n-3}n!(5n^3+30n^2+53n+32)-\frac{(2n-1)!}{(n-1)!2^{n-1}}(n+1)(2n^2+7n+4).
\end{eqnarray*}

Adding $P_1$, $P_2$, $P_3$, and $P_4$, we have
\begin{equation*}
\resizebox{0.98\linewidth}{!}{$\displaystyle
\begin{aligned}
C_8={}&
2^{n-7}n!
\Bigl(
5n^6+75n^5+479n^4+1641n^3
+3164n^2+3148n+1280
\Bigr)
\\
&-\frac{(2n-1)!}{3465(n-1)!2^{n-1}}
\Bigl(
320n^6+4968n^5+30896n^4+96828n^3
+160709n^2+125544n+34650
\Bigr).
\end{aligned}
$}
\end{equation*}
\end{proof}

\begin{lemma}\label{lem39}
    Let $n \geq 1$. The total number of networks in ${\cal P}_{n,4}$ associated with the ninth group of component graphs in Figure \ref{fig:42} is given by:
\begin{eqnarray*}
C_9&=&\frac{1}{3}\cdot2^{n-8}(n+1)!(21n^5+312n^4+1979n^3+6444n^2+10060n+6144)\\
        &&-\frac{(2n+1)!}{3465\cdot n!2^{n-1}}(96n^5+1328n^4+7419n^3+20599n^2+27153n+13860).
    \end{eqnarray*}
\end{lemma}

\begin{proof}
    A network associated with this group of component graphs can be divided into top and bottom subnetworks. 
    
    The top subnetwork is a one-component network with one reticulation, $j\ge 0$ network leaves and one designated leaf $\ell$, which provides $|{\cal OCP}_{j+2,1}|$ possibilities.
    
   The bottom subnetwork is a network with two reticulations in which parallel
edges are allowed. To reconstruct the original network, the designated leaf
$\ell$ is inserted into an edge of the bottom subnetwork. Distinct insertion
edges, however, do not always produce distinct networks. Here, an
automorphism of a bottom subnetwork is understood to preserve its root and
all labeled network leaves.

Suppose that the two reticulations $r_1$ and $r_2$ have the same two parents
$p_1$ and $p_2$, and that $p_1$ and $p_2$ have a common parent $p$, as shown
in Figure~\ref{4_ret_fig_11}. Exchanging $p_1$ and $p_2$ defines an
automorphism of the bottom subnetwork. This automorphism exchanges the two
edges within each of the following pairs:
\[
 \{(p,p_1),(p,p_2)\},\qquad
 \{(p_1,r_1),(p_2,r_1)\},\qquad
 \{(p_1,r_2),(p_2,r_2)\}.
\]
Consequently, inserting $\ell$ into either edge of any one of these pairs
produces the same network.

Conversely, a direct case analysis of the bottom subnetworks associated with
the ninth group shows that an automorphism can map one insertion edge to a
distinct insertion edge only in the local configuration described above.
Therefore, the configuration displayed in
Figure~\ref{4_ret_fig_11} is the only one that gives rise to
indistinguishable insertion edges.

We now count the bottom subnetworks without parallel edges that
contain this local symmetric configuration. If the bottom subnetwork contains $n-j$ network leaves, we
first construct a tree with $n-j$ leaves and then replace one of its $(n-j-1)$ tree nodes by the circled configuration in Figure~\ref{4_ret_fig_11}. For $n-j\geq2$, we have
\[
S_{\mathrm{low}}(n-j)
=t_{n-j}(n-j-1)
=\frac{(2n-2j-2)!}{2^{n-j-1}(n-j-2)!}.
\]
When $n-j=1$, we have
\[
S_{\mathrm{low}}(1)=0.
\]

Therefore, among bottom subnetworks without parallel edges, the
number of those not containing this symmetric configuration is
\[
A_{\mathrm{low}}(n-j)
=|\mathcal{P}_{n-j,2}|-S_{\mathrm{low}}(n-j).
\]
In particular,
\[
A_{\mathrm{low}}(1)=|\mathcal{P}_{1,2}|=1.
\]
For $n-j\geq2$, this gives
\begin{align*}
A_{\mathrm{low}}(n-j)
={}&\frac{(2n-2j-2)!}
{3\cdot2^{n-j-1}(n-j-1)!}
\bigl[6(n-j)^4+19(n-j)^3+18(n-j)^2\\
&\qquad{}-7(n-j)-3\bigr]\\
&-2^{n-j-1}(n-j+1)!(2n-2j+3).
\end{align*}
    
    \begin{figure}[t!]
\centering
\includegraphics[width=0.2\linewidth]{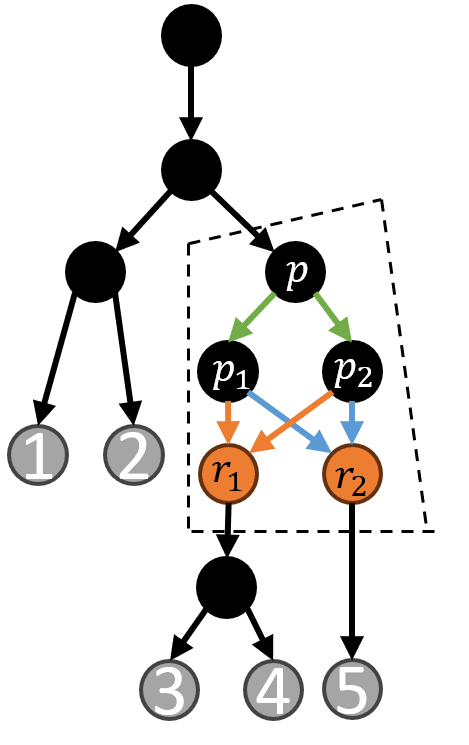}
\caption{\label{4_ret_fig_11} \bf A symmetric topology of the bottom subnetwork for a network associated with a component graph in the ninth group \textit{before} inserting the designated leaf.} In the circled structure, the edges with the same color are indistinguishable.
\end{figure}

If the subnetwork below contains a symmetric topology, there are $2n-2j+2$ inequivalent choices of insertion edge for the designated leaf $\ell$. Thus, the number of networks corresponding to this case is
\begin{eqnarray*}
    P_1&=&\sum_{j=0}^{n-2} \binom{n}{j} |{\cal OCP}_{j+2,1}| S_{low}(n-j)\times (2n-2j+2)\\
    &=&\sum_{j=0}^{n-2} \frac{n!}{(n-j)!j!} \frac{(2j+2)!}{2^{j+1}j!} \frac{(2n-2j-3)![(n-j)^2-1]}{(n-j-2)!2^{n-j-3}}\\
    &=& \frac{n!}{2^{n-1}}\sum_{j=0}^{n-2}\binom{2j}{j}\binom{2n-2j}{n-j} \frac{(2j+1)(j+1)(n-j-1)(n-j+1)}{2n-2j-1}\\
    &=& 2^{n-3}n!(n^3+5n^2+26n+16)-\frac{(2n+1)!}{2^{n-1}n!}(n+1).
\end{eqnarray*}

If the bottom subnetwork has no parallel edges and does not contain
the symmetric configuration mentioned above (see Figure \ref{4_ret_fig_12}a as an example), then there are
$2n-2j+5$ inequivalent choices of insertion edges for the designated leaf
$\ell$. Substituting
$A_{\mathrm{low}}(n-j)
=|\mathcal{P}_{n-j,2}|-S_{\mathrm{low}}(n-j)$
and separating the contribution of $S_{\mathrm{low}}(n-j)$, which
vanishes when $j=n-1$, we obtain
\begin{eqnarray*}
    P_2&=&\sum_{j=0}^{n-1} \binom{n}{j} |{\cal OCP}_{j+2,1}| A_{low}(n-j) \times(2n-2j+5)\\
    &=&\frac{n!}{3\cdot2^{n}}\sum_{j=0}^{n-1}\binom{2j}{j}\binom{2n-2j}{n-j} \frac{(2j+1)(j+1)(2n-2j+5)}{2n-2j-1}\\
    &&\times\Big[6(n-j)^4+19(n-j)^3+18(n-j)^2-4(n-j)-6\Big]\\
    &&-n!2^{n-1}\sum_{j=0}^{n-1}\binom{2j}{j}\frac{(2j+1)(j+1)(n-j+1)(2n-2j+3)(2n-2j+5)}{2^{2j}}\\
    &&-\frac{n!}{2^n}\sum_{j=0}^{n-2}\binom{2j}{j}\binom{2n-2j}{n-j}\frac{(2j+1)(j+1)}{2n-2j-1}[2(n-j)^2+3n-3j-5]\\
    &=& \frac{1}{3}\cdot2^{n-8}n!(n+5)(21n^5+228n^4+1079n^3+2500n^2+2252n+768)\\
    &&-\frac{(2n+1)!}{3465\cdot n!2^n}(192n^5+2656n^4+14442n^3+38690n^2+42195n+17325).
\end{eqnarray*}

    \begin{figure}[t!]
\centering
\includegraphics[width=0.5\linewidth]{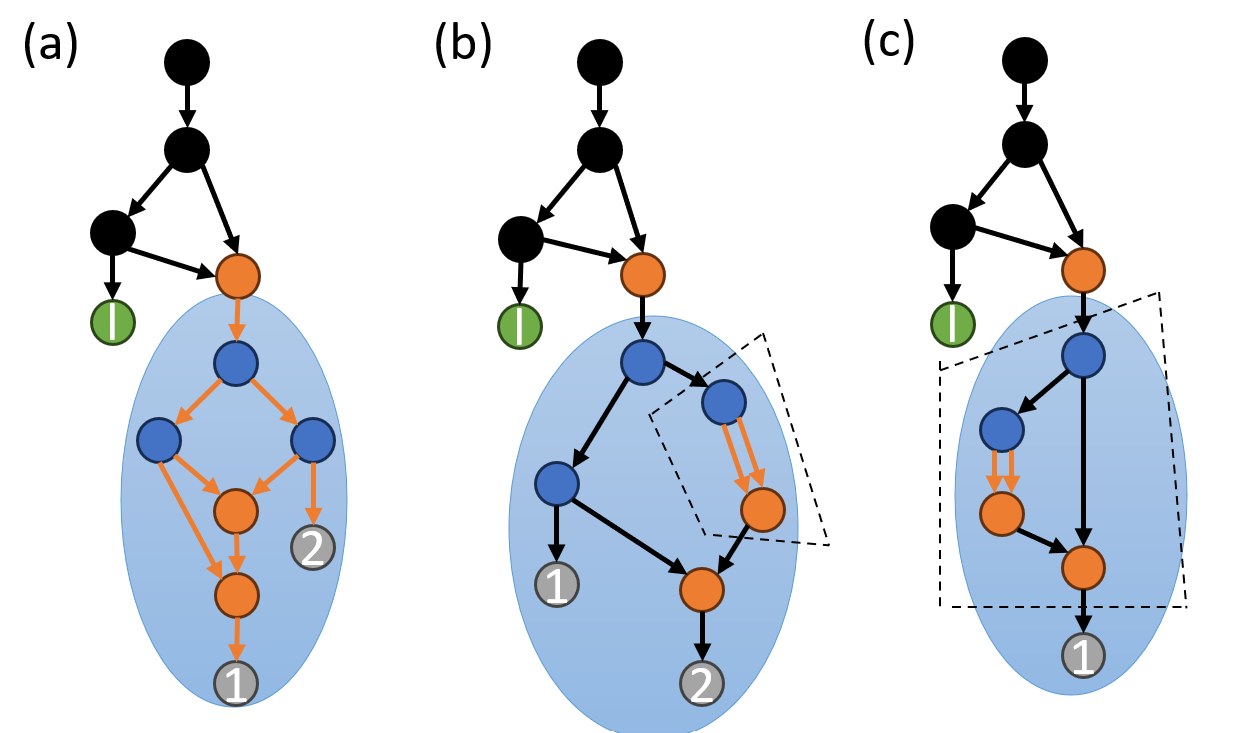}
\caption{\label{4_ret_fig_12} \bf Possible asymmetric topologies for a network with a component graph in the ninth group \textit{before} inserting the designated leaf.} The green leaf $\ell$ denotes the designated leaf in the top subnetwork. The shaded parts denote the bottom subnetwork before inserting $\ell$. The orange edges denote the possible edges into which $\ell$ can be inserted. (a) The bottom subnetwork is a network with two reticulations. (b) The bottom subnetwork is obtained by inserting a pair of parallel edges into a network with one reticulation. (c) The bottom subnetwork is obtained by inserting a double-parallel edge structure into a tree.
\end{figure}

If the subnetwork below contains a pair of parallel edges, it can be obtained by replacing an edge from a network with one reticulation with the circled parallel edge structure in Figure \ref{4_ret_fig_12}b. Then, $\ell$ is inserted into one of the parallel edges. Consequently, the number of networks corresponding to this case is
\begin{eqnarray*}
P_3&=&\sum_{j=0}^{n-2}\binom{n}{j} |{\cal OCP}_{j+2,1}|\times|{\cal P}_{n-j,1}|\times(2n-2j+2)\\
&=&\sum_{j=0}^{n-2} \frac{n!}{(n-j)!j!} \frac{(2j+2)!}{2^{j+1}j!} \Big[\frac{(n-j)(2n-2j)!}{2^{n-j}(n-j)!}-2^{n-j-1}(n-j)!   \Big]\times(2n-2j+2)\\
&=& \frac{n!}{2^{n-1}}\sum_{j=0}^{n-2}\binom{2j}{j}\binom{2n-2j}{n-j} (2j+1)(j+1)(n-j)(n-j+1)\\
&&-n!2^n \sum_{j=0}^{n-2} \binom{2j}{j} \frac{(2j+1)(j+1)(n-j+1)}{2^{2j}}\\
&=&2^{n-5}(n+1)!n(3n^2+15n+14)-\frac{(2n+1)!}{105\cdot (n-1)! 2^{n-2}}(3n^2+19n+13).
\end{eqnarray*}

Also, we can consider a network containing a double parallel-edge structure. Such networks can be obtained by replacing an edge from a tree with the circled structure in Figure \ref{4_ret_fig_12}c. Then, $\ell$ is inserted into one of the parallel edges. Therefore, the number of networks corresponding to this case is
\begin{eqnarray*}
    P_4 &=& \sum_{j=0}^{n-1}\binom{n}{j} |{\cal OCP}_{j+2,1}|\times t_{n-j}\times(2n-2j-1)\\
    &=& \sum_{j=0}^{n-1} \frac{n!}{(n-j)!j!} \frac{(2j+2)!}{2^{j+1}j!} \frac{(2n-2j-1)!}{2^{n-j-1}(n-j-1)!}\\
    &=& \frac{n!}{2^{n}} \sum_{j=0}^{n-1} \binom{2j}{j} \binom{2n-2j}{n-j} (2j+1)(j+1)\\
    &=&2^{n-2}(n+1)!(3n+4)-\frac{(2n+1)!}{n!2^n}(n+1).
\end{eqnarray*}

    Adding $P_1$, $P_2$, $P_3$, and $P_4$, we have
    \begin{eqnarray*}
        C_9&=&\frac{1}{3}\cdot2^{n-8}(n+1)!(21n^5+312n^4+1979n^3+6444n^2+10060n+6144)\\
        &&-\frac{(2n+1)!}{3465\cdot n!2^{n-1}}(96n^5+1328n^4+7419n^3+20599n^2+27153n+13860).
    \end{eqnarray*}

\end{proof}

\begin{lemma}
\label{lem310}
    Let $n \geq 2$. The total number of networks in ${\cal P}_{n,4}$ associated with the tenth group of component graphs in Figure \ref{fig:43} is given by:
    \begin{equation*}
\resizebox{0.98\linewidth}{!}{$\displaystyle
\begin{aligned}
C_{10}={}&
\frac{1}{3}2^{n-9}n!
\left(
21n^6+287n^5+1585n^4+4193n^3
+5066n^2+2864n+576
\right)\\
&+\frac{9n(2n-2)!}{2^{n-1}(n-2)!}\\
&-\frac{n(2n-3)!}{10395\,(n-2)!2^{n-1}}
\left(
3072n^6+31144n^5+119292n^4+196774n^3
+96249n^2+187945n-239466
\right).
\end{aligned}
$}
\end{equation*}
\end{lemma}

\begin{proof}
    A network associated with this group of component graphs can be divided into top and bottom subnetworks. The top subnetwork is a one-component network with one reticulation. The bottom subnetwork is a network with three reticulations. See Figure \ref{4_ret_fig_13} for illustration.

    First, assume the bottom subnetwork contains more than one network leaf. Let  $P_1$ denote the number of such networks. Then
    {\small
\begin{align*}
P_1
&=
\begin{aligned}[t]
\sum_{j=1}^{n-2}
\binom{n}{j}
\lvert{\cal OCP}_{j+1,1}\rvert
\lvert{\cal P}_{n-j,3}\rvert
\end{aligned}
\\
&=
\begin{aligned}[t]
&\sum_{j=1}^{n-2}
\frac{n!}{j!(n-j)!}
\frac{(2j)!}{2^j(j-1)!}
\Bigg\{
\frac{(2n-2j-2)!}
     {3(n-j-1)!2^{n-j}}
\\
&\qquad {}\times
\Big[
8(n-j)^6+88(n-j)^5+366(n-j)^4
+640(n-j)^3
\\
&\hspace{10mm}
{}+325(n-j)^2-155(n-j)-114
\Big]
\\
&\qquad {}
-\frac{1}{3}(n-j+1)!2^{n-j-4}
\Big[
48(n-j)^3+367(n-j)^2
\\
&\hspace{10mm}
{}+959(n-j)+840
\Big]
\Bigg\}
\end{aligned}
\\
&=
\begin{aligned}[t]
&\frac{n!}{3\cdot 2^{n+1}}
\sum_{j=1}^{n-2}
\binom{2j}{j}
\binom{2n-2j}{n-j}
\frac{j}{2n-2j-1}
\\
&\qquad {}\times
\Big[
8(n-j)^6+88(n-j)^5+366(n-j)^4
+640(n-j)^3
\\
&\hspace{10mm}
{}+325(n-j)^2-155(n-j)-114
\Big]
\\
&\quad {}
-\frac{1}{3}n!2^{n-4}
\sum_{j=1}^{n-2}
\binom{2j}{j}
\frac{j(n-j+1)}{2^{2j}}
\\
&\qquad {}\times
\Big[
48(n-j)^3+367(n-j)^2
+959(n-j)+840
\Big]
\end{aligned}
\\
&=
\begin{aligned}[t]
&\frac{1}{3}2^{n-9}n!
\Big(
21n^6+287n^5+1585n^4+4193n^3
\\
&\hspace{10mm}
{}+5066n^2+2864n+576
\Big)
\\
&\quad {}
-\frac{n(2n-3)!}
      {10395\,(n-2)!2^{n-1}}
\Big(
3072n^6+31144n^5+119292n^4
\\
&\hspace{10mm}
{}+196774n^3+96249n^2
+187945n-239466
\Big).
\end{aligned}
\end{align*}
}

    \begin{figure}[t!]
\centering
\includegraphics[width=0.1\linewidth]{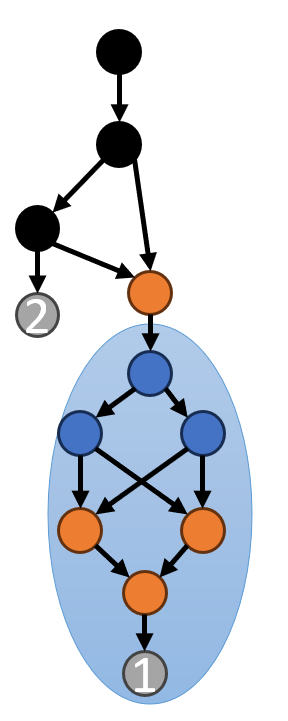}
\caption{\label{4_ret_fig_13} \bf Illustration of the network decomposition in the proof of Lemma \ref{lem310}.} The unshaded subnetwork is the top subnetwork, a one-component network with one reticulation. The shaded subnetwork is the bottom subnetwork, a network with three reticulations.
\end{figure}

    Second, assume the bottom subnetwork contains only one leaf. Let $P_2$ denote the number of such networks. Then
    \begin{eqnarray*}
        P_2&=&\binom{n}{n-1} \times|{\cal OCP}_{n,1}|\times|{\cal P}_{1,3}|\\
        &=&\frac{9n(2n-2)!}{2^{n-1}(n-2)!}.
    \end{eqnarray*}

    By adding $P_1$ and $P_2$ together, we complete the proof.
\end{proof}

Combining the preceding group counts gives the following result.

\begin{theorem}\label{thm31}
    For $n\ge 2$, the number of binary phylogenetic networks with four
reticulations on \(n\) taxa is given by
    \begin{equation*}
\resizebox{0.98\linewidth}{!}{$\displaystyle
\begin{aligned}
\lvert{\cal P}_{n,4}\rvert
={}&
\frac{(2n-4)!}{189\,(n-2)!2^{n-1}}
\Big(
504n^9+10836n^8+91414n^7+362607n^6+568813n^5
\\
&\hspace{10mm}
{}-326256n^4-1950557n^3-1378566n^2
+440523n+367416
\Big)
\\
&-\frac{1}{3}(n+1)!2^{n-4}
\Big(
32n^5+558n^4+3901n^3
+13523n^2+23008n+15264
\Big).
\end{aligned}
$}
\end{equation*}

    For $n=1$, $|{\cal P}_{1,4}|=109.$

\end{theorem}
\begin{proof}
For \(n\geq4\), summing the expressions \(C_1,\ldots,C_{10}\) from
Lemmas~\ref{lem31}--\ref{lem310} and simplifying gives the stated formula.

When \(n=3\), the first group is empty, so \(C_1=0\). When \(n=2\), the
first two groups are empty, so \(C_1=C_2=0\). Direct substitution shows that
the same closed-form expression remains valid in both cases.

Finally, when \(n=1\), only the fifth through ninth groups contribute.
Lemmas~\ref{lem35}--\ref{lem39} give
\[
C_5=45,\qquad C_6=31,\qquad C_7=3,\qquad C_8=22,\qquad C_9=8.
\]
Therefore,
\[
|{\cal P}_{1,4}|=45+31+3+22+8=109.
\]
\end{proof}

\begin{table}[b!]
\centering
\small

\begin{tabular}{c|rrrr}
\hline
\textbf{$k\setminus n$}
& \textbf{1}
& \textbf{2}
& \textbf{3}
& \textbf{4}\\
\hline
\textbf{0} & 1   & 1     & 3       & 15\\
\textbf{1} & 0   & 2     & 21      & 228\\
\textbf{2} & 1   & 18    & 279     & 4,530\\
\textbf{3} & 9   & 225   & 4,980   & 110,205\\
\textbf{4} & 109 & 3,881 & 113,424 & 3,212,190\\
\hline
\end{tabular}

\par\medskip

\begin{tabular}{c|rrrr}
\hline
\textbf{$k\setminus n$}
& \textbf{5}
& \textbf{6}
& \textbf{7}
& \textbf{8}\\
\hline
\textbf{0}
& 105
& 945
& 10,395
& 135,135\\
\textbf{1}
& 2,805
& 39,330
& 623,385
& 11,055,240\\
\textbf{2}
& 79,455
& 1,517,670
& 31,568,355
& 712,879,650\\
\textbf{3}
& 2,540,955
& 61,951,680
& 1,605,463,650
& 44,266,135,095\\
\textbf{4}
& 92,486,235
& 2,759,501,745
& 86,034,598,650
& 2,812,696,803,000\\
\hline
\end{tabular}

\par\medskip

\begin{tabular}{c|rr}
\hline
\textbf{$k\setminus n$}
& \textbf{9}
& \textbf{10}\\
\hline
\textbf{0}
& 2,027,025
& 34,459,425\\
\textbf{1}
& 217,237,545
& 4,689,345,150\\
\textbf{2}
& 17,405,723,475
& 457,518,565,350\\
\textbf{3}
& 1,297,410,689,925
& 40,350,659,129,550\\
\textbf{4}
& 96,530,594,780,025
& 3,477,318,517,324,575\\
\hline
\end{tabular}

\par\medskip

\begin{tabular}{c|rr}
\hline
\textbf{$k\setminus n$}
& \textbf{11}
& \textbf{12}\\
\hline
\textbf{0}
& 654,729,075
& 13,749,310,575\\
\textbf{1}
& 110,367,613,125
& 2,813,814,441,900\\
\textbf{2}
& 12,892,904,830,575
& 387,986,437,688,850\\
\textbf{3}
& 1,328,802,652,069,200
& 46,229,827,577,558,625\\
\textbf{4}
& 131,376,539,868,027,300
& 5,199,820,531,562,106,450\\
\hline
\end{tabular}

\caption{The numbers of binary phylogenetic networks with up to four
reticulations on $n$ taxa, for $1\leq n\leq12$.}
\label{tab_count}
\end{table}



\section{Conclusion}\label{section4}

In this paper, we studied phylogenetic networks with four reticulations. Using the component graph method, we obtained an explicit counting formula for this class of networks. In particular, Theorem~\ref{thm31} makes it possible to compute the number of such networks for every number of leaves \(n\). Selected values for networks with zero to four reticulations are listed in Table~\ref{tab_count}. These results extend the enumeration of networks with a small fixed number of reticulations and provide a further example of how structural decompositions can be used to derive exact counting formulas.

The method also suggests that component graphs are a useful tool for organizing the combinatorial structure of phylogenetic networks. The results of this work may contribute to a better understanding of the full space of phylogenetic networks. Exact enumeration formulas can provide benchmarks for computational methods, support the design of sampling algorithms, and offer useful combinatorial information for future studies of network-based evolutionary models. 



\clearpage

\section*{Statements and Declarations}

\noindent\textbf{Funding.}
This work was supported by the Ministry of Education, Singapore,
under its Academic Research Fund Tier 1
(Grant No.\ A-8001951-00-00).

\medskip

\noindent\textbf{Competing Interests.}
The authors have no relevant financial or non-financial interests
to disclose.

\medskip

\noindent\textbf{Data Availability.}
No datasets were generated or analyzed during the current study.

\medskip

\noindent\textbf{Author Contributions.}
Hao Yu and Louxin Zhang contributed to the conception and design
of the study. Hao Yu performed the combinatorial analysis and prepared
the original draft. Louxin Zhang contributed to the methodology,
supervision, and revision of the manuscript.


\newpage
\section*{Appendix}
\subsection*{Proof of Lemma~\ref{lem1}}
(1) Identities~(\ref{eq4}),~(\ref{eq5}),~(\ref{eq6}), and~(\ref{eq7}) appear in    Example 4.3 in \cite{Koshy_book},
Identity 2.1 in \cite{diekema2022combinatorial}, 
Identities 1.38 and 
1.35 in \cite{Gould_book}, respectively. They can be proved by induction on $n$.

(2) To prove the identity~(\ref{eq8}), we define $G(x)=\frac{1}{\sqrt{1-4x}}=\displaystyle\sum_{k=0}^{\infty}\binom{2k}{k}x^k.$ 

Note that  $$[x^n]\left[G(x)\cdot x^{t+1}\cdot\frac{d^{t+1}}{dx^{t+1}}G(x) \right]=\displaystyle\sum_{k=0}^{n}\binom{2k}{k}\binom{2n-2k}{n-k}k(k-1)\cdots(k-t)$$
and  $$\frac{d^{t+1}}{dx^{t+1}}G(x)=4^{t+1} \frac{(2t+1)!!}{2^{t+1}}(1-4x)^{-3/2-t}.$$ Therefore, we have
$$[x^n]\left[G(x)\cdot x^{t+1}\cdot\frac{d^{t+1}}{dx^{t+1}}G(x) \right]=\frac{4^{t+1}(2t+1)!!}{2^{t+1}}[x^{n-t-1}](1-4x)^{-2-t}.$$

Note that by the binomial theorem, we have $[x^{n-t-1}](1-4x)^{-2-t}=\binom{n}{t+1}4^{n-t-1}.$ Therefore, we have $$[x^n]\left[G(x)\cdot x^{t+1}\cdot\frac{d^{t+1}}{dx^{t+1}}G(x) \right]=\frac{(2t+1)!!}{2^{t+1}}\binom{n}{t+1}4^n.$$
This concludes the proof.

To prove Identity~\((\ref{eq9})\), let
\[
f(x)=\frac{1}{\sqrt{1-x}}
=
\sum_{k=0}^{\infty}\binom{2k}{k}\frac{x^k}{2^{2k}}.
\]
By differentiating repeatedly, we have
\[
f^{(t+1)}(x)
=
\sum_{k=t+1}^{\infty}
\binom{2k}{k}
\frac{k(k-1)\cdots(k-t)}{2^{2k}}
x^{k-t-1}.
\]
On the other hand,
\[
f^{(t+1)}(x)
=
\frac{(2t+1)!!}{2^{t+1}}(1-x)^{-t-\frac32}.
\]
Hence, if
\[
A(x):=
\sum_{k=0}^{\infty}
\binom{2k}{k}
\frac{k(k-1)\cdots(k-t)}{2^{2k}}x^k,
\]
then
\[
A(x)
=
x^{t+1}f^{(t+1)}(x)
=
\frac{(2t+1)!!}{2^{t+1}}
x^{t+1}(1-x)^{-t-\frac32}.
\]

Therefore,
\[
\sum_{k=0}^n
\binom{2k}{k}
\frac{k(k-1)\cdots(k-t)}{2^{2k}}
=
[x^n]\frac{A(x)}{1-x}.
\]
Thus
\[
\begin{aligned}
\sum_{k=0}^n
\binom{2k}{k}
\frac{k(k-1)\cdots(k-t)}{2^{2k}}
&=
\frac{(2t+1)!!}{2^{t+1}}
[x^{n-t-1}](1-x)^{-t-\frac52}  \\
&=
\frac{(2t+1)!!}{2^{t+1}}
\binom{n+\frac12}{n-t-1}.
\end{aligned}
\]
Using
\[
\binom{n+\frac12}{n-t-1}
=
\frac{\Gamma(n+\frac32)}
{\Gamma(n-t)\Gamma(t+\frac52)},
\]
we get
\[
\sum_{k=0}^n
\binom{2k}{k}
\frac{k(k-1)\cdots(k-t)}{2^{2k}}
=
\frac{(2t+1)!!}{2^{t+1}}
\frac{\Gamma(n+\frac32)}
{\Gamma(n-t)\Gamma(t+\frac52)}.
\]
Now
\[
\Gamma\left(t+\frac52\right)
=
\frac{(2t+3)!!}{2^{t+2}}\sqrt{\pi}
\]
and
\[
\Gamma\left(n+\frac32\right)
=
\frac{(2n+1)!}{2^{2n+1}n!}\sqrt{\pi},
\]
where $\Gamma(1/2)=\sqrt{\pi}$ is used.
Therefore,
\[
\begin{aligned}
\sum_{k=0}^n
\binom{2k}{k}
\frac{k(k-1)\cdots(k-t)}{2^{2k}}
&=
\frac{(2t+1)!!}{2^{t+1}}
\frac{1}{(n-t-1)!}
\frac{(2n+1)!}{2^{2n+1}n!}
\frac{2^{t+2}}{(2t+3)!!}  \\
&=
\frac{(2n+1)!}
{2^{2n}n!(n-t-1)!(2t+3)}.
\end{aligned}
\]
This proves Identity~\((\ref{eq9})\).

(3) To prove Identity~(\ref{eq10}), let
\[
S_n(t):=\sum_{k=0}^{n}
\binom{2k}{k}\binom{2n-2k}{n-k}\frac{1}{2k-t}.
\]

We first rewrite the summand in hypergeometric form. Recall that
\[
\binom{2k}{k}=4^k\frac{(\frac12)_k}{k!},
\]
where $(x)_k=x(x+1)...(x+k-1)$ and $(x)_0=1$.

Moreover,
\[
\binom{2n-2k}{n-k}
=
4^{n-k}\frac{(\frac12)_{n-k}}{(n-k)!}
=
4^{n-k}\frac{(\frac12)_n}{n!}
\frac{(-n)_k}{(\frac12-n)_k}.
\]
Hence
\[
\binom{2k}{k}\binom{2n-2k}{n-k}
=
\binom{2n}{n}
\frac{(\frac12)_k(-n)_k}{(\frac12-n)_k k!}.
\]
Also, for every \(k\),
\[
\frac{1}{2k-t}
=-\frac{1}{t}\frac{-t}{2k-t}=-\frac{1}{t}\frac{-\frac t2}{k-\frac t2}=
-\frac{1}{t}
\frac{(-\frac t2)_k}{(1-\frac t2)_k}.
\]
Therefore,
\[
S_n(t)
=
-\frac{1}{t}\binom{2n}{n}
\sum_{k=0}^{n}
\frac{(-n)_k(-\frac t2)_k(\frac12)_k}
{(\frac12-n)_k(1-\frac t2)_k}
\frac{1}{k!}.
\]
Equivalently,
\[
S_n(t)
=
-\frac{1}{t}\binom{2n}{n}
{}_3F_2
\left(
\begin{matrix}
-n,\ -\frac t2,\ \frac12\\
\frac12-n,\ 1-\frac t2
\end{matrix}
;1
\right).
\]

Now we apply Saalsch\"utz's Theorem
\[
{}_3F_2
\left(
\begin{matrix}
-n,\ a,\ b\\
c,\ 1+a+b-c-n
\end{matrix}
;1
\right)
=
\frac{(c-a)_n(c-b)_n}{(c)_n(c-a-b)_n}.
\]
Taking
\[
a=-\frac t2,\qquad b=\frac12,\qquad c=\frac12-n,
\]
we have
\[
{}_3F_2
\left(
\begin{matrix}
-n,\ -\frac t2,\ \frac12\\
\frac12-n,\ 1-\frac t2
\end{matrix}
;1
\right)
=
\frac{
(\frac{t+1}{2}-n)_n(-n)_n
}{
(\frac12-n)_n(\frac t2-n)_n
}.
\]
Consequently
\[
S_n(t)
=
-\frac{1}{t}\binom{2n}{n}
\frac{
(\frac{t+1}{2}-n)_n(-n)_n
}{
(\frac12-n)_n(\frac t2-n)_n
}.
\]

Since \(t\) is odd and \(1\leq t\leq 2n-1\), we may write
\[
t=2s+1,\qquad 0\leq s\leq n-1.
\]
Then
\[
\left(\frac{t+1}{2}-n\right)_n
=
(s+1-n)_n.
\]
This rising factorial contains the factor \(0\), because
\[
s+1-n+(n-s-1)=0.
\]
Hence
\[
\left(\frac{t+1}{2}-n\right)_n=0.
\]
On the other hand, the denominator factors
\[
\left(\frac12-n\right)_n
\quad\text{and}\quad
\left(\frac t2-n\right)_n
\]
are products of nonzero half-integers, so they are nonzero. It follows that
\[
S_n(t)=0.
\]
This proves Identity~\((\ref{eq10})\).

\end{document}